\documentclass[12pt, onecolumn,a4paper,unpublished]{quantumarticle}
\pdfoutput=1 
\usepackage[T1]{fontenc}
\usepackage{makecell}
\usepackage{tikz}
\usetikzlibrary{positioning, calc}
\usetikzlibrary{arrows}
\usepackage{tikz-3dplot}
\usepackage{graphicx}
\usepackage{mdwlist}
\usepackage{color}
\usepackage{thmtools}
\usepackage{amssymb}
\usepackage{physics}
\usepackage{amsmath}
\usepackage{array}
\usepackage{doi}
\usepackage{url,hyperref}
\usepackage[capitalize, nameinlink,noabbrev]{cleveref}
\hypersetup{colorlinks={true},linkcolor={blue},citecolor=red}

\usepackage[numbers,sort&compress]{natbib}
\usetikzlibrary{fadings}
\usetikzlibrary{decorations.pathmorphing}
\usepackage{mathrsfs}
\usepackage{authblk}
\usepackage[mathscr]{euscript}
\usepackage{enumitem}
\usepackage{quantikz}
\usepackage[normalem]{ulem}

\newcommand{\pub}{\mathsf{pub}}

\newcommand{\Rent}{\mathsf{R}\|^*}
\newcommand{\Rsim}{\mathsf{R}\|}
\newcommand{\Qent}
{\mathsf{Q}\|^*}
\newcommand{\Qsim}{\mathsf{Q}\|}

\newcommand{\vabs}[1]{\left\Vert #1 \right\Vert}

\DeclareMathAlphabet\mathbfcal{OMS}{cmsy}{b}{n}

\usetikzlibrary{decorations.pathreplacing}
\tikzset{snake it/.style={decorate, decoration=snake}}

\usetikzlibrary{decorations.pathreplacing,decorations.markings}

\tikzset{
    >=stealth',
    punkt/.style={
           rectangle,
           rounded corners,
           draw=black, very thick,
           text width=6.5em,
           minimum height=2em,
           text centered},
    pil/.style={
           ->,
           thick,
           shorten <=2pt,
           shorten >=2pt,},
  on each segment/.style={
    decorate,
    decoration={
      show path construction,
      moveto code={},
      lineto code={
        \path [#1]
        (\tikzinputsegmentfirst) -- (\tikzinputsegmentlast);
      },
      curveto code={
        \path [#1] (\tikzinputsegmentfirst)
        .. controls
        (\tikzinputsegmentsupporta) and (\tikzinputsegmentsupportb)
        ..
        (\tikzinputsegmentlast);
      },
      closepath code={
        \path [#1]
        (\tikzinputsegmentfirst) -- (\tikzinputsegmentlast);
      },
    },
  },
  mid arrow/.style={postaction={decorate,decoration={
        markings,
        mark=at position .5 with {\arrow[#1]{stealth'}}
      }}}
}

\usepackage{slashed}
\usepackage{float} 
\usepackage{subcaption}

\mathchardef\mhyphen="2D

\newcommand{\CDQS}{\mathsf{CDQS}}
\newcommand{\pp}{\mathsf{pp}}
\newcommand{\pc}{\mathsf{pc}}

\newcommand{\PSM}{\mathsf{PSM}}
\newcommand{\PSQM}{\mathsf{PSQM}}

\newcommand{\eps}{\epsilon}

\newtheorem{theorem}{Theorem}

\newtheorem{corollary}[theorem]{Corollary}

\newtheorem{definition}[theorem]{Definition}

\newtheorem{lemma}[theorem]{Lemma}

\newenvironment{proof}[1][Proof]{\noindent\textbf{#1. }}{\ \rule{0.5em}{0.5em}}
\usepackage{cancel}

\newcommand{\ntot}{n_{\text{tot}}}
\begin{document}

\author[1]{Uma Girish}
\email{ug2150@columbia.edu}
\orcid{0000-0003-3055-9406}

\author[2,3]{Alex May}
\email{amay@perimeterinstitute.ca}
\orcid{0000-0002-4030-5410}

\author[1]{Natalie Parham}
\email{natalie@cs.columbia.edu}
\orcid{0000-0002-8792-1229}

\author[1]{Henry Yuen}
\email{hyuen@cs.columbia.edu}
\orcid{0000-0002-2684-1129}

\affiliation[1]{Columbia University}
\affiliation[2]{Perimeter Institute for Theoretical Physics}
\affiliation[3]{Institute for Quantum Computing, University of Waterloo}

\title{New bounds on private simultaneous quantum message passing}

\abstract{In the private simultaneous message (PSM) setting, $k$ players obtain inputs $x_i\in\{0,1\}^n$ and then independently send messages to a referee, who should learn $f(x_1,...,x_k)$ but no other information about $(x_1,...,x_k)$. 
The PSM setting was introduced as a minimal model for secure multiparty computation.
In the quantum setting, PSM has been related to non-local quantum computation (NLQC), and has several connections to the complexity of Boolean functions.
The communication and correlation cost of implementing private simultaneous message passing may be much larger than the cost without privacy, and the cost of privacy in this setting remains poorly understood. 
Here, we give new upper and lower bounds on the PSM model, in both the quantum and classical settings. 
Concretely, we prove two lower bounds:
\begin{itemize}
    \item Ne\v{c}iporuk's measure lower bounds the entanglement required for $k$-player quantum PSM with perfect correctness. This can be evaluated to give quadratic lower bounds for some explicit functions.
    \item The rank of the communication matrix of $f(x_1,x_2)$ lower bounds 2-player quantum PSM with perfect privacy but imperfect correctness. This implies a previously unknown lower bound on classical PSM with imperfect correctness.
\end{itemize}
When allowing both quantum communication and shared entanglement, these two bounds are the first lower bounds on quantum PSM that make use of the privacy condition.
Regarding upper bounds, we show:
\begin{itemize}
    \item Letting $s$ be the size of a quantum circuit computing $f$, $d_f$ be the circuit depth, $k$ the number of players, $n$ the number of bits received by each player, and $\epsilon$ the correctness parameter of the PSM protocol, we obtain the upper bound $\PSM_k^*(f) \leq (kn +s) \cdot \log^{O( d_f)}(s/\eps)$. 
    \item The square of the Fourier 1 norm of $f$, $\Vert \hat{f}\Vert_1^2$, upper bounds the classical PSM complexity, $\PSM(f)\leq O(\Vert \hat{f} \Vert^2_1)$.
\end{itemize}
In proving the first upper bound, we also generalize existing $T$-depth based techniques for NLQC from $2$ to $k\geq 2$ parties, and consider cases where the Clifford layers are restricted to having small light cones. 
These generalizations may be of independent interest. 
}

\maketitle

\tableofcontents

\section{Introduction}

What is the cost of information-theoretic privacy? 
This is a fundamental question in cryptography which reappears across many settings, and which is closely related to understanding the complexity of Boolean functions. 
In this work, we make progress on understanding this problem in both the classical and quantum settings by giving several new upper and lower bounds on the private simultaneous message (PSM) model \cite{feige1994minimal,ishai1997private}.

The PSM setting is illustrated in \cref{fig:quantumandclassicalPSM}.
The setting involves $k$ parties, who we label player 1, player 2, etc, and a referee. 
Player $i$ receives input $x_i\in\{0,1\}^n$.
All parties including the referee agree in advance on a choice of Boolean function $f:\{0,1\}^{kn}\rightarrow \{0,1\}$.
Classically, the players share randomness that is inaccessible to the referee,  and send classical messages to the referee.
The goal is for the referee to compute $f(x_1,...,x_k)$ without learning anything further about the value of $(x_1,...,x_k)$.
We consider this model in both quantum and classical variants.
When $n=1$, so that each of $k$ players receives a single bit of input, this setting is also an example of a \emph{decomposable randomized encoding}.
Another special case we often consider fixes $k=2$. 

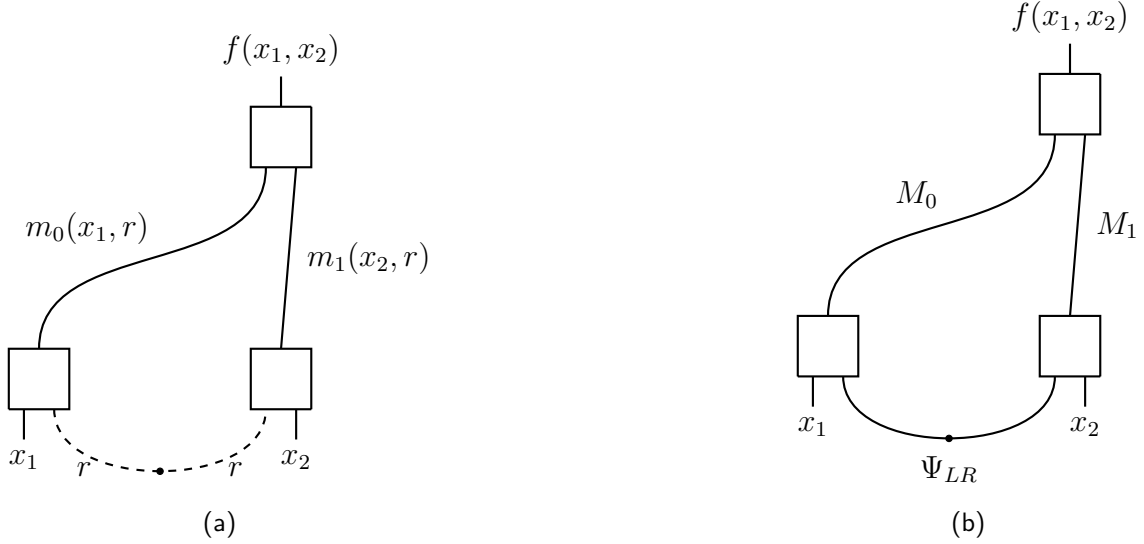
\begin{figure*}
    \centering
    \begin{subfigure}{0.45\textwidth}
    \centering
    \begin{tikzpicture}[scale=0.4]
    
    \draw[thick] (-5,-5) -- (-5,-3) -- (-3,-3) -- (-3,-5) -- (-5,-5);
    
    \draw[thick] (5,-5) -- (5,-3) -- (3,-3) -- (3,-5) -- (5,-5);
    
    \draw[thick] (5,5) -- (5,3) -- (3,3) -- (3,5) -- (5,5);
    
    \draw[thick] (4,-3) -- (4.5,3);
    
    \draw[thick] (-4,-3) to [out=90,in=-90] (3.5,3);
    
    \draw[thick,dashed] (-3.5,-5) to [out=-90,in=-90] (3.5,-5);
    \node[below] at (-2.5,-6.4) {$r$};
    \draw[black] plot [mark=*, mark size=3] coordinates{(0,-7.05)};
    \node[below] at (2.5,-6.4) {$r$};

    \node[left] at (0,1) {$m_0(x_1,r)$};
    \node[right] at (4.5,0) {$m_1(x_2,r)$};
    
    \draw[thick] (-4.5,-6) -- (-4.5,-5);
    \node[below] at (-4.5,-6) {$x_1$};
    
    \draw[thick] (4.5,-6) -- (4.5,-5);
    \node[below] at (4.5,-6) {$x_2$};
    
    \draw[thick] (4,5) -- (4,6);
    \node[above] at (4,6) {$f(x_1,x_2)$};
    
    \end{tikzpicture}
    \caption{}
    \label{fig:PSM}
    \end{subfigure}
    \hfill
    \begin{subfigure}{0.45\textwidth}
    \centering
    \begin{tikzpicture}[scale=0.4]
    
    \draw[thick] (-5,-5) -- (-5,-3) -- (-3,-3) -- (-3,-5) -- (-5,-5);
    
    \draw[thick] (5,-5) -- (5,-3) -- (3,-3) -- (3,-5) -- (5,-5);
    
    \draw[thick] (5,5) -- (5,3) -- (3,3) -- (3,5) -- (5,5);
    
    \draw[thick] (4,-3) -- (4.5,3);
    
    \draw[thick] (-4,-3) to [out=90,in=-90] (3.5,3);
    
    \draw[thick] (-3.5,-5) to [out=-90,in=-90] (3.5,-5);
    \node[below] at (0,-7.25) {$\Psi_{LR}$};
    \draw[black] plot [mark=*, mark size=3] coordinates{(0,-7.05)};

    \node[left] at (0,1) {$M_0$};
    \node[right] at (4.5,0) {$M_1$};
    
    \draw[thick] (-4.5,-6) -- (-4.5,-5);
    \node[below] at (-4.5,-6) {$x_1$};
    
    \draw[thick] (4.5,-6) -- (4.5,-5);
    \node[below] at (4.5,-6) {$x_2$};
    
    \draw[thick] (4,5) -- (4,6);
    \node[above] at (4,6) {$f(x_1,x_2)$};
    
    \end{tikzpicture}
    \caption{}
    \label{fig:PSQM}
    \end{subfigure}
    \caption{A private simultaneous message protocol (PSM). We show the case with two players for simplicity, but in general we consider $k\geq 2$ players. Players 1 and 2 do not communicate. Player $i$ holds input $x_i\in\{0,1\}^n$. The referee should be able to learn $f(x_1,...,x_k)$ but nothing else about $(x_1,...,x_k)$. a) In the classical setting, the players share a random string $r$ and send classical messages to the referee. b) In the quantum setting the randomness is replaced with an entangled state $\Psi_{LR}$, and the messages can be quantum.}
    \label{fig:quantumandclassicalPSM}
\end{figure*}

\subsection{Related work}

The PSM model was introduced in a classical context in \cite{feige1994minimal} as a simple toy model for secure multi-party computation. 
Since then, several applications and connections to other primitives in cryptography have appeared \cite{ishai1997private}. 

The definition of PSM is information-theoretic and does not refer to complexity theory. 
Nonetheless, a relationship between PSM and complexity emerges: known protocols for implementing PSM have a communication and randomness cost set by the complexity of the function $f$.
For instance, in \cite{feige1994minimal} an efficient protocol for functions in $\mathsf{NL}$ was given, based on a reduction to a PSM protocol for group products. 
This was improved in \cite{ishai1997private} to give efficient protocols for functions in $\mathsf{Mod}_p\mathsf{L}$\footnote{$\mathsf{Mod}_p\mathsf{L}$ contains $\mathsf{NL}$, and this containment is believed to be strict.} as well as other `counting' variants of log-space classes. 

Without the privacy condition, PSM complexity would be at most linear, since Alice and Bob can simply send their inputs to the referee. 
Meanwhile, most notions of Boolean function complexity can be as large as exponential, so it is clear that without privacy PSM complexity and Boolean function complexity can only be loosely related. 
With privacy however, the relationship may be much tighter. 
Better understanding the relationship between privacy and the complexity of Boolean functions is one motivation for studying the cost of privacy in PSM. 

Only a few lower bound techniques that make use of the privacy condition are known for PSM.  
In the classical setting, one such technique was proposed in \cite{feige1994minimal} and then corrected in \cite{applebaum2020communication}. 
This technique, when we have $k=2$ players, gives a lower bound in terms of a condition on rectangles in the communication matrix for $f$; for random functions, it leads to a $3n-o(1)$ lower bound. 
This can be seen to make non-trivial use of privacy in that without privacy, the cost is at most $2n$, the total input length. 

In the classical context, a lower bound on $k$-party PSM was proven in \cite{ball2020complexity,ball2022note}. 
This lower bound is in terms of Ne\v{c}iporuk's measure, an object previously known to lower bound formula size for Boolean functions. 
Concretely, \cite{ball2020complexity,ball2022note} prove that for perfectly correct and perfectly secure PSM,
\begin{align}
    \PSM(f)\geq \frac{1}{2}G^*(f)
\end{align}
where $G^*(f)$ is Ne\v{c}iporuk's measure. 
A similar bound holds with perfect security relaxed.
Heuristically, Ne\v{c}iporuk's measure can be understood as the logarithm of the number of distinct functions that appear when fixing subsets of the variables. 
In more detail, we sum over a partition of the variables and consider the number of distinct functions that arise when fixing all but the variables in the current subset of the partition. 
For some explicit functions on $k$ inputs each of length $n$, Ne\v{c}iporuk's measure evaluates to $\frac{k^2n}{2\log(kn)}$, so when (for instance) $n=1$, this is a nearly quadratic lower bound. 
This bound also necessarily uses privacy, since it is super-linear.

Recently \cite{kawachi2021communication,allerstorfer2024relating} a new perspective on PSM and related primitives has emerged, which considers quantum variants of these settings.
This provides a new setting in which to explore the relationship between privacy, complexity, and communication cost. 
In \cite{kawachi2021communication}, quantum PSM was first studied. 
In that context, with $k=2$ players, a lower bound of $3n-o(1)$ for random functions was proven in the context of quantum communication but restricting to classical shared randomness.
In \cite{kawachi2021communication} it was shown that for some relations, communication cost in PSM can be exponentially smaller when allowing shared entanglement; this was later proven for a partial function in \cite{girish2025comparing}. 

Quantum PSM exhibits a relationship between privacy and notions of quantum complexity. 
In particular, in \cite{allerstorfer2024relating} it was pointed out that communication and entanglement cost in PSM with shared entanglement allowed is upper bounded in terms of the $T$-depth of any unitary that computes the relevant Boolean function. 
Later \cite{girish2025magic}, similar techniques were used to observe that certain communication complexity protocols can be transformed into PSM protocols.
In particular, let $\PSM^*(f)$ denote the communication cost of PSM when allowing shared entanglement but restricting to classical messages, and let $\Qsim^*$ be the communication cost in the simultaneous message model, allowing quantum communication and shared entanglement. 
Then \cite{girish2025magic} showed that
\begin{align}
    \PSM^*(f) \leq \left(\Qsim^*(f)+a\right)^{d_T}
\end{align}
where $d_T$ is the $T$-depth of the unitary applied by the referee in the $\Qsim^*$ protocol, and $a$ is the number of qubits of ancilla used by the referee. 
In practice, this leads to new efficient PSM$^*$ protocols constructed from existing $\Qsim^*$ protocols in several interesting cases, and in particular leads to new separations between communication complexity classes. Furthermore, this shows that a separation between $\PSM^*$ and $\Qsim^*$ implies a T-depth lower bound. 

\subsection{Our results}

We contribute to the understanding of the cost of privacy in quantum and classical PSM by proving two new lower bounds on entanglement cost, and two new upper bounds.
One of each of our lower and upper bounds is new to even classical PSM.    
Our lower bounds are the first lower bounds on fully quantum PSM which make use of the privacy condition, where by fully quantum we mean that both entanglement and quantum communication are allowed. 
Previously, lower bounds that used privacy only applied to the case where the messages are quantum, but no shared entanglement is allowed \cite{kawachi2021communication}.

Our first lower bound is in terms of the rank of the communication matrix of the target function $f$, and applies to two player quantum PSM,
\begin{align}
    \boxed{\pp\overline{\PSQM}^*(f)\geq \frac{1}{4}\log \rank f - \frac{1}{4}.}
\end{align}
Here, the left hand side denotes the entanglement cost of fully quantum PSM (allowing entanglement and quantum messages), and requiring perfect privacy. 
The communication matrix $M_f$ of $f$ is a $2^n\times 2^n$ matrix with entries $f(x,y)$.
We compute the rank over the complex numbers. 
Notice that this lower bound follows without using privacy if we assume perfect correctness, because quantum communication complexity with perfect correctness is lower bounded by the log-rank \cite{buhrman2001communication}. 
Our result holds with imperfect correctness however, and relies instead on the perfect privacy condition. 
Because $\pp\overline{\PSM}(f) \geq \pp\overline{\PSQM}^*(f)$, we obtain the same lower bound on perfectly private classical PSM. 
This result appears to be new to the classical setting, but is proven from the quantum perspective. 

Our second lower bound applies to $k$-player quantum PSM, with perfect correctness. 
Then, denoting the total entanglement shared among all players by $\pc\overline{\PSQM}_k^*(f)$, we find that
\begin{align}
    \boxed{\pc\overline{\PSQM}_k^*(f) \geq \frac{1}{2}\alpha_\delta G^*(f)-\beta_\delta.}
\end{align}
Here the object $G^*(f)$ appearing in the lower bound is \emph{Ne\v{c}iporuk's measure}, and the $\alpha_\delta, \beta_\delta$ are functions of the security parameter $\delta$ for the protocol. 
Because Ne\v{c}iporuk's measure can be quadratic in $k$, and hence larger than the input size, it is clear that this bound also relies on the privacy condition. 

For upper bounds, we prove two new results. 
First, we prove an upper bound in the $k$ player setting based on a Clifford+$T$ decomposition of any circuit computing the function $f$. 
We find the upper bound
\begin{align}
    \PSM_k^*(f)\leq O((K\ell)^{d_T-1}\cdot kn \cdot \min\{k,\ell\})
\end{align}
where $d_T$ is the $T$-depth of the circuit computing $f$, $k$ is the number of players, $n$ is the number of input bits per player, and $\ell$ measures how much any single Clifford layer spreads the support of an operator. 
This result generalizes the $T$-depth upper bound of \cite{speelman2015instantaneous}, to allow for $k>2$ players and restricted Clifford layers. 

As an application of this upper bound, we show that functions with low-depth quantum circuits \cite{yao1993quantum} can be computed efficiently in the $\PSM^*$ model. In particular, we consider circuits with gates that only act on a constant number of qubits. 
Letting $d_f$ be the minimal depth of any quantum circuit computing $f$, we find
\begin{align}
    \boxed{\PSM^*_k(f) \leq (kn + s) \cdot \log^{O(d_f)}(s/\eps).}
\end{align} 
Thus we find that functions computable in depth $\log(nk)/ \log\log(nk)$ can be computed by a polynomial-cost $\PSM^*_k$ protocol. 
It is interesting to compare this with the classical seting: a combination of Barrington's theorem~\cite{barrington1986bounded} and the PSM construction of~\cite{feige1994minimal} for branching programs implies that there are polynomial-cost classical PSM protocols for all functions computable by logarithmic-depth (classical) circuits. 
We leave it as an interesting open question of improving our bound to get polynomial-cost $\PSM^*$ protocols for quantum circuits with depth $\Omega(\log n)$.

Our second upper bound is in terms of the Fourier 1 norm of $f$, denoted $\Vert \hat{f}\Vert_1$. 
Recall that we define the Fourier transform of $f$ by defining functions $\chi_S(x)=(-1)^{S\cdot x}$ where $S$ labels a subset of the bits of $x$, and then expressing $f(x)$ as
\begin{align}
    f(x) = \sum_S \hat{f}(S) \chi_S(x).
\end{align}
The Fourier 1 norm is then
\begin{align}
    \Vert\hat{f}\Vert_1 = \sum_S |\hat{f}(S)|.
\end{align}
The Fourier 1 norm appears in several contexts \cite{o2014analysis}; most relevantly \cite{grolmusz1997power} proved a $O(\Vert \hat{f}\Vert^2_1)$ upper bound on the simultaneous message passing (SMP) model (which is the same as PSM, but without privacy imposed). This was then used to prove a number of circuit lower bounds. 
We upgrade this upper bound on the SMP model to the PSM model, proving
\begin{align}
    \boxed{\PSM(f) \leq O(\Vert \hat{f}\Vert^2_1).}
\end{align}
It may be interesting to explore applications of this upper bound to proving circuit lower bounds. 

\paragraph{Acknowledgments.} Research at the Perimeter Institute is supported by the Government of Canada through the Department of Innovation, Science and Industry Canada and by the Province of Ontario through the Ministry of Colleges and Universities. UG, NP, and HY are supported by AFOSR award FA9550-23-1-0363, NSF awards CCF-2530159, CCF-2144219, and CCF-2329939, and by the Sloan Foundation. NP is supported by the Google PhD Fellowship.

\section{Model definitions and quantum information tools}

\subsection{Quantum information tools}

Unless otherwise noted, by $\log$ we always mean the base 2 logarithm. 
We define the von Neumann entropy by
\begin{align}
    S(A)_\rho = -\tr \left(\rho_A \log \rho_A\right).
\end{align}
We define the mutual information as
\begin{align}
    I(A:B)_\rho = S(A)_\rho + S(B)_\rho - S(AB)_\rho.
\end{align}
We make use of the following continuity property of the mutual information, which follows from the Alicki-Fannes-Winter inequality \cite{alicki2004continuity,winter2016tight}.
\begin{theorem}\label{thm:MIcontinuity}
Suppose that $||\sigma-\rho||_1\leq 2\epsilon$, and define
\begin{align}
    h(\epsilon)= (\epsilon+1) \log(\epsilon+1) -\epsilon \log \epsilon.
\end{align}
Then
\begin{align}
    |I(A:B)_\sigma - I(A:B)_\rho| \leq 3\epsilon \log d_{A} + 2h(\epsilon).
\end{align}
\end{theorem}

\subsection{Communication models}

In this section we define the relevant classical and quantum communication models. 
In all cases we are interested in the simultaneous message passing scenario, so that the communication pattern allows only for the players to send messages to a referee. 
The models we consider can then be divided into models without privacy, and those with privacy.
We begin by defining the non-private models. 

\begin{definition}\label{def:SMP}
Let $f : \{0,1\}^n\times \{0,1\}^n\rightarrow \{0,1\}$ be a (partial or total) Boolean function, and $\epsilon \in [0,1]$ be a parameter. 
A \textbf{simultaneous message passing} protocol $P$ for $f$ involves three parties, Alice, Bob, and a referee. 
Alice receives $x\in \{0,1\}^n$ as input and Bob receives $y\in \{0,1\}^n$. Alice and Bob send the referee (quantum or classical) message systems $M_A$ and $M_B$ respectively, and the referee subsequently outputs a bit $c=P(x,y)$.

\textbf{Messages.} The messages that Alice and Bob send to the referee can be quantum, denoted by $\Qsim$ or classical, denoted by $\Rsim$. 

\textbf{Correctness.} The protocol is $\epsilon$-correct if for all $(x,y)$ in the support of $f$,
\begin{align*}
    \Pr[P(x,y)=f(x,y)] \geq 1-\epsilon \enspace.
\end{align*}
By default we assume $\epsilon=1/3$. 
Focusing on the case when Alice and Bob send classical messages and $\epsilon=0$, we obtain the deterministic model of classical simultaneous communication, denoted by $\mathsf{D}\|$.

\textbf{Cost of a protocol.} The cost of the protocol denoted by $\mathrm{cost}(P)$ is defined to be the total number of bits (resp. qubits) sent by Alice and Bob in the $\Rsim$ (resp. $\Qsim$) model. The $\Rsim_\epsilon$ complexity of $f$ is defined as follows
\begin{equation*}
    \Rsim_{\epsilon}(f) = \min_{P: P \text{ is $\epsilon$-correct}}\mathrm{cost}(P)
\end{equation*}
and the $\Qsim_\epsilon$ complexity is analogously defined. 

\textbf{Randomness.} Alice and Bob typically have private randomness, but we also consider a variation of the simultaneous message model where we allow public randomness. 
In particular, we allow all three players (Alice, Bob and the referee) to hold a shared random string $r$ of arbitrary length. 
They can then use $r$ as an input to their local operations. 
We label the cost to compute $f$ $\epsilon$-correctly in this model by $\Rsim^{\pub}_\epsilon(f)$ (resp. $\Qsim^{\pub}_\epsilon(f)$) when the messages are classical (resp. quantum).  

\textbf{Entanglement.} We may allow Alice and Bob to share entanglement, denoted by the superscript $*$ and resulting in the models $\Qent$ and $\Rent$ depending on whether the messages to the referee are quantum or classical. 
\end{definition}

Next we take up the private variations of these models. 

\begin{definition}\label{def:PSQM}
    A \textbf{private simultaneous message} task is defined by a choice of (partial or total) Boolean function $f:\{0,1\}^n\times \{0,1\}^n\rightarrow \{0,1\}$. Let $\epsilon, \delta \in [0,1]$ be parameters.
    The inputs to the task are $n$-bit strings $x$ and $y$ given to Alice and Bob, respectively.
    Alice then sends a message system $M_0$ to the referee, and Bob sends a message system $M_1$. 
    From the combined message system $M=M_0M_1$, the referee prepares an output bit $z$ whose system is denoted by $Z$.  
    We require the task be completed in a way that satisfies the following two properties.
    \begin{itemize}
        \item \textbf{$\epsilon$-correctness:} There exists a decoding map $\mathbfcal{V}_{M \rightarrow Z}$ such that, for all $(x,y)$ in the support of $f$, 
        \begin{align}
            \left \|\mathbfcal{V}_{M \rightarrow Z} \left(\rho_{M}(x,y)\right) - \ketbra{f_{x,y}}{f_{x,y}}_Z\right \|_1 \leq \epsilon
        \end{align}
        where $\rho_M(x,y)$ is the density matrix on $M$ produced on inputs $x,y$ and $f_{x,y}=f(x,y)$.
        \item \textbf{$\delta$-security:} There exists a simulator, which is a quantum channel $\mathbfcal{S}_{Z\rightarrow M}(\cdot)$, such that for all $(x,y)$ on which $f$ is defined 
        \begin{align}
            \left \|\rho_{M}(x,y) - \mathbfcal{S}_{Z\rightarrow M}(\ketbra{f_{x,y}}{f_{x,y}}_Z)\right \|_1 \leq \delta.
        \end{align}
        Stated differently, the state of the message systems is $\delta$-close to one that depends only on the function value, for every choice of input.
    \end{itemize}
   
    \textbf{Messages.} When the messages are quantum, we will refer to this model as $\PSQM$ and when the messages are classical, we refer to the model by $\PSM$.

    \textbf{Entanglement.} When Alice and Bob share entanglement, we denote it by the superscript $*$, obtaining the model $\PSQM^*$ when Alice and Bob send quantum messages and $\PSM^*$ when Alice and Bob send classical messages. 

    \textbf{Communication cost of a protocol.} The communication cost of the protocol is defined to be the total number of bits sent by Alice and Bob in the $\PSM$ or $\PSM^*$ models. 
    We denote the minimal cost over all $\epsilon=1/3$ correct, $\delta=1/3$ secure protocols by $\PSM(f)$ or $\PSM^*(f)$. 
    The communication cost measures $\PSQM(f)$ and $\PSQM^*(f)$ are defined similarly, now counting qubits of communication. 

    \textbf{Correlation cost of a protocol.} The correlation cost of the protocol is defined to be the log dimension of the quantum state shared among Alice and Bob at the start of the protocol, which counts both classical randomness and shared entanglement. We use an overline to denote the correlation cost in various models, $\overline{\PSM}(f)$, $\overline{\PSQM}(f)$, etc.
\end{definition}

If we enforce that the protocol is perfectly correct we add the suffix $\pc$; if we enforce that the protocol is perfectly private we add $\pp$. 
Thus for example the communication cost of perfectly correct PSQM$^*$ protocol for function $f$ is $\pc\PSQM^*(f)$. 

We also consider PSM settings where the input is split among $k$ parties, rather than two. 
The definition is similar to the above, but we replace the input space $X\times Y$ with $X_1\times ...\times X_k$, and have messages computed separately from each of the inputs. 
In this case we add a subscript $k$ to the model designation, for instance $\PSQM^*_k$ is $k$-player PSM with quantum communication and shared entanglement.
The $k$ players may share a $k$-partite entangled state.
To define the entanglement cost in this case, suppose that a protocol uses resource state ${\Psi}_{E_1E_2...E_k}$. 
Then we define the correlation cost to be the log dimension of this state. 
Note that we allow arbitrary isometries in the PSQM$^*$ protocol, so the correlation cost measure doesn't count any ancilla used.

\section{Lower bound from Ne\v{c}iporuk's measure}\label{sec:Neclowerbound}

In this section we develop a lower bound technique on correlation cost in the PSQM$^*$ model, in the setting where we give each of $k$ parties $n$ bits of the input. 
We start by developing the definition of Ne\v{c}iporuk’s measure, which counts the number of distinct functions that appear when restricting a given function to a subset of its inputs. 
Our strategy is adapted from an analogous classical setting \cite{ball2022note}. 

\subsection{Ne\v{c}iporuk's measure}

The variant of Ne\v{c}iporuk’s measure that appears in our lower bounds involves the following notion of a restriction of a $k$ input function. 

\begin{definition}\textbf{(First-bit restriction)}
For any function $f : \{\{0, 1\}^n\}^k \rightarrow \{0, 1\}$ and any set $S \subseteq [k]$, the first-bit restriction of $f$ to $S$ using $(\alpha,\beta)$ is the function 
\begin{align}
    f_{S|(\alpha,\beta)} : \{0, 1\}^{|S|} \rightarrow \{0, 1\}
\end{align}
defined by restricting the inputs as follows:
\begin{itemize}
    \item Fix the $n$-bit inputs corresponding to $\bar{S}$ to $\alpha\in \{\{0,1\}^n\}^{|\bar{S}|}$.
    \item Fix the last $n-1$ bits of each $n$-bit input corresponding to $S$ to the values described by $\beta\in\{\{0,1\}^{n-1}\}^{|S|}$.
\end{itemize}
\end{definition}

Now we can define the modified Nečiporuk’s measure. 

\begin{definition}\textbf{(Modified Nečiporuk’s measure)} 
Let $f : \{\{0, 1\}^n\}^k \rightarrow \{0, 1\}$ be a function.
For any subset $S \subseteq [k]$, define 
\begin{align}
    g^*_S(f) := \max_{\beta\in\{\{0,1\}^{n-1}\}^{|S|}} \log|\{f_{S|(\alpha,\beta)} : \alpha \in \{\{0, 1\}^n\}^{|\bar{S}|}, f_{S|(\alpha,\beta)} \not\equiv 0\}|.
\end{align}
For any positive integer $m \leq k$, let $V = (V_1, V_2, ..., V_m)$ denote an $m$-partition of $[k]$.
The Nečiporuk measure is defined to be $G^*(f) := \max_V\sum_{V_i\in V} g^*_{V_i}(f)$.
\end{definition}

We will eventually prove a lower bound on the $k$ party PSQM$^*$ complexity in terms of the Ne\v{c}iporuk measure. 
First, we need the following claim, which restricts the number of orthogonal states on the $AB$ Hilbert space assuming that all of the density matrices on $B$ (nearly) agree. 

\begin{lemma}\label{lemma:subsystemsizebound}
     Call $d_X$ the number of states $\{\rho^i_{AB}\}_i$ satisfying 
     \begin{align}
         F(\rho^i_{AB}, \rho^j_{AB}) = \delta_{ij} \,\,\,\,\,\,\text{and}\,\,\,\,\,\, \Vert \rho_{B}^i-\rho_B\Vert_1 \leq \delta
     \end{align}
     with $\rho_B$ a single fixed state. Then
     \begin{align}
         \log d_X \leq \frac{1}{1-3\delta/2}\left[2\log d_A + 2h(\delta/2) \right].
     \end{align}
     where $h(x)=(1+x)\log (1+x) - x\log x$.
\end{lemma}
\begin{proof}
Consider the states
\begin{align}
    \sigma_{XAB} &= \frac{1}{d_X} \sum_{i=1}^{d_X} \ketbra{i}{i}_X\otimes \rho_{AB}^i \nonumber \\
    \tilde{\sigma}_{XB} &= \frac{1}{d_X} \sum_{i=1}^{d_X} \ketbra{i}{i}_X\otimes \rho_{B}
\end{align}
Then we have
\begin{align}
    \Vert \sigma_{XB} - \tilde{\sigma}_{XB}\Vert_1 =\frac{1}{{d_X}}\sum_i \Vert \rho_{B}^i-\rho_B\Vert_1 \leq \delta
\end{align}
Using \cref{thm:MIcontinuity} and that $I(X:B)_{\tilde{\sigma}}=0$, we can bound the mutual information $I(X:B)_\sigma$,
\begin{align}
    I(X:B)_\sigma = I(X:B)_\sigma - I(X:B)_{\tilde{\sigma}} \leq \frac{3}{2}\delta \log d_X + 2h(\delta/2).
\end{align}
But also, since the $\rho_{AB}^i$ all have orthogonal support, we can measure $AB$ and determine $i$ so then also $I(AB:X)=\log d_X$.
But then we use
\begin{align}
    \log d_X &= I(AB:X)_{\sigma} \nonumber \\
            &= I(A:X|B)_{\sigma} + I(B:X)_{\sigma} \nonumber \\
            &\leq 2S(A)_{\sigma} + I(B:X)_{\sigma} \nonumber \\
            &\leq 2S(A) + \frac{3}{2}\delta \log d_X + 2h(\delta/2) \nonumber \\
            &\leq 2 \log d_A + \frac{3}{2}\delta \log d_X + 2h(\delta/2)
\end{align}
which can be re-arranged to the desired inequality.
\end{proof}

For intuition, consider the case where $\delta=0$ and $A$, $B$ are both qubits.
Then the bound says that any set of orthogonal states on $AB$ which all have the same marginal on $A$ must be of size at most $4$. 
We can saturate this bound by for instance choosing the four Bell states on $AB$, which are orthogonal but all have $\rho_B=\mathcal{I}_B/2$. 

\subsection{Lower bound from Nečiporuk’s measure}

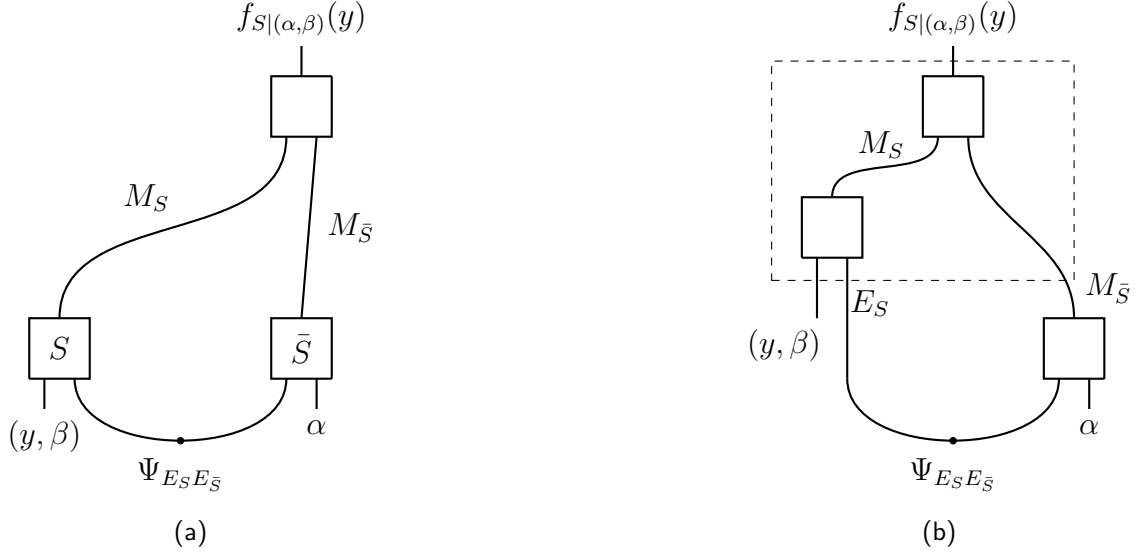
\begin{figure*}
    \centering
    \begin{subfigure}{0.45\textwidth}
    \centering
    \begin{tikzpicture}[scale=0.4]
    
    \draw[thick] (-5,-5) -- (-5,-3) -- (-3,-3) -- (-3,-5) -- (-5,-5);
    \node at (-4,-4) {$S$};
    
    \draw[thick] (5,-5) -- (5,-3) -- (3,-3) -- (3,-5) -- (5,-5);
    \node at (4,-4) {$\bar{S}$};
    
    \draw[thick] (5,5) -- (5,3) -- (3,3) -- (3,5) -- (5,5);
    
    \draw[thick] (4,-3) -- (4.5,3);
    
    \draw[thick] (-4,-3) to [out=90,in=-90] (3.5,3);
    
    \draw[thick] (-3.5,-5) to [out=-90,in=-90] (3.5,-5);
    \node[below] at (0,-7.25) {$\Psi_{E_SE_{\bar{S}}}$};
    \draw[black] plot [mark=*, mark size=3] coordinates{(0,-7.05)};

    \node[left] at (0,1) {$M_S$};
    \node[right] at (4.5,0) {$M_{\bar{S}}$};
    
    \draw[thick] (-4.5,-6) -- (-4.5,-5);
    \node[below] at (-4.5,-6) {$(y,\beta)$};
    
    \draw[thick] (4.5,-6) -- (4.5,-5);
    \node[below] at (4.5,-6) {$\alpha$};
    
    \draw[thick] (4,5) -- (4,6);
    \node[above] at (4,6) {$f_{S|(\alpha,\beta)}(y)$};
    
    \end{tikzpicture}
    \caption{}
    \end{subfigure}
    \hfill
    \begin{subfigure}{0.45\textwidth}
    \centering
    \begin{tikzpicture}[scale=0.4]
    
    \draw[thick] (-5,-1) -- (-5,1) -- (-3,1) -- (-3,-1) -- (-5,-1);
    
    \draw[thick] (5,-5) -- (5,-3) -- (3,-3) -- (3,-5) -- (5,-5);
    
    \draw[thick] (1,5) -- (1,3) -- (-1,3) -- (-1,5) -- (1,5);
    
    \draw[thick] (4,-3) to [out=90,in=-90] (0.5,3);
    
    \draw[thick] (-4,1) to [out=90,in=-90] (-0.5,3);
    
    \draw[thick] (-3.5,-5) to [out=-90,in=-90] (3.5,-5);
    \draw[thick] (-3.5,-5) to (-3.5,-1);
    \node[below] at (0,-7.25) {$\Psi_{E_SE_{\bar{S}}}$};
    \draw[black] plot [mark=*, mark size=3] coordinates{(0,-7.05)};

    \node[left] at (-1.25,2.75) {$M_{S}$};
    \node[right] at (4,-2) {$M_{\bar{S}}$};
    
    \draw[thick] (4.5,-6) -- (4.5,-5);
    \node[below] at (4.5,-6) {$\alpha$};

    \node at (-2.75,-2.5) {$E_S$};
    \draw[thick] (-4.5,-3) -- (-4.5,-1);
    \node[below left] at (-4,-3) {$(y,\beta)$};
    
    \draw[thick] (0,5) -- (0,6);
    \node[above] at (0,6) {$f_{S|(\alpha,\beta)}(y)$};

    \draw[dashed] (-6,-1.75) -- (4,-1.75) -- (4,5.5) -- (-6,5.5) -- (-6,-1.75);
    
    \end{tikzpicture}
    \caption{}
    \end{subfigure}
    \caption{a) A $\PSQM^*$ protocol. The $n$ players have been divided into two subsets, $S$ and $\bar{S}$. The input to the $\bar{S}$ players is set to $\alpha$. The last $n-1$ bits of each of the players in $S$ is set to $\beta$, while the first bit inputs are free, and the players can choose any string $y$ to take as input. Correctness of the $\PSQM^*$ protocol gives that the referee can compute $f_{S|(\alpha,\beta)}(y)$ from the message systems $M_S$, $M_{\bar{S}}$. b) In the proof of \cref{thm:Neciporuklowerbound}, we observe that given $E_SM_{\bar{S}}$, $f_{S|(\alpha,\beta)}(y)$ can be computed for any value of $y$. By computing this reversibly, this can be repeated for all values of $y$. This means the function $f_{S|(\alpha,\beta)}$ is determined by $E_SM_{\bar{S}}$. }
    \label{fig:Neciporukproof}
\end{figure*}

Finally, we are ready to prove the quantum PSM lower bound. 
The overall strategy is as follows. 
We partition the inputs into $S$ and $\bar{S}$. 
The inputs to $\bar{S}$ are fixed to a string $\alpha\in \{\{0,1\}^{n}\}^{|\bar{S}|}$. 
The last $n-1$ bits of each input to the players in $S$ is fixed to $\beta\in\{\{0,1\}^{(n-1)}\}^{|S|}$, while the remaining inputs are taken to be $y\in\{0,1\}^{|S|}$. 
We observe that the message $M_{\bar{S}}$ from players in $\bar{S}$, along with the initial entanglement on players in $S$, call it $E_S$, can be used to compute $f_{S|(\alpha,\beta)}(y)$ for any choice of $y$.
See \cref{fig:Neciporukproof}. 
Further, when the protocol is perfectly correct, we can copy out the value of $y$, uncompute, and run forward again with a new choice of input. 
Thus in fact the state on $S\bar{S}$ determines all values of $f_{S|(\alpha,\beta)}$ and hence determines $\alpha$. 
Naively, that means $E_SM_{\bar{S}}$ needs to have a dimension lower bounded by the number of distinct functions. 
Alone, this isn't strong enough to get our bound, and instead we want to lower bound the dimension of $E_S$ alone. 
To achieve this, we use security to show that all the states on $E_SM_{\bar{S}}$ have density matrices that agree on $M_{\bar{S}}$, and then apply \cref{lemma:subsystemsizebound} to bound the dimension of $E_S$ alone. 

\begin{theorem}\label{thm:Neciporuklowerbound}
    Consider a perfectly correct PSQM$^*$ protocol for function $f$ with security parameter $\delta>0$. 
    Then, the correlation cost is lower bounded by
    \begin{align}
        \pc\overline{\PSQM}^*(f) \geq \frac{1}{2}(1-3\delta/2)G^*(f)- h(\delta/2)p
    \end{align}
    where $p$ is the number of subsets in the partition $V^*$ that maximizes the sum $\sum_{V_i\in V}g_{V_i}^*(f)$.
\end{theorem}
\begin{proof}
    Consider a subset of players $S$ and the complement set $\bar{S}$.
    Considering the modified Ne\v{c}iporuk measure, fix $\beta$ to its maximizing value. 
    Then choose a set of distinct $\alpha$ that each lead to distinct first-bit restrictions $f_{S|(\alpha,\beta)}$.  
    Denote the remaining inputs to $f_{S|(\alpha,\beta)}$ as $y$. 
    Let the message system sent by $S$ be $M_S$, and the message system sent by $\bar{S}$ be $M_{\bar{S}}$. 
    Similarly, let the entangled resource systems held by $S$ and $\bar{S}$, before they act with their local operators, be $E_S$ and $E_{\bar{S}}$ respectively. 
    
    We claim that all the possible message density matrices (for each distinct input $y$) agree on $M_{\bar{S}}$, up to an error set by $\delta$. 
    This is because for all $\alpha, \alpha'$ the functions we get by restriction are still non-zero by assumption, so there exist choices of $y,y'$ such that $f_{S|(\alpha,\beta)}(y)=f_{S|(\alpha',\beta)}(y')$, so then by privacy 
    \begin{align}
        \Vert\rho_{M_SM_{\bar{S}}}(y,\alpha,\beta)-\rho_{M_SM_{\bar{S}}}(y',\alpha',\beta)\Vert_1\leq \delta.
    \end{align}
    But by causality the reduced density matrix on $M_{\bar{S}}$ can't depend on $y$, so we get
    \begin{align}
        \Vert\rho_{M_{\bar{S}}}(\alpha)-\rho_{M_{\bar{S}}}(\alpha')\Vert_1 \leq \delta
    \end{align}
    as claimed. 
    Define $\rho_{M_{\bar{S}}}(\alpha_0)\equiv \rho_{M_{\bar{S}}}$ for some fixed $\alpha_0$, so all $\rho_{M_{\bar{S}}}(\alpha)$ are $\delta$ close to a fixed density matrix.

    Next, we show that $\alpha$ can be determined from $E_S M_{\bar{S}}$, in other words we show that $F(\rho_{E_SM_{\bar{S}}}(\alpha,\beta), \rho_{E_SM_{\bar{S}}}(\alpha',\beta))=\delta_{\alpha,\alpha'}$. 
    To see this, consider a unitary extension of the operations of each of the $S$ players acting on $E_S$ and input $I$, followed by the referee's operations, to define a single unitary taking in $E_SM_{\bar{S}}$ and a copy of the inputs and producing
    \begin{align}
        U_{IE_SM_{\bar{S}}\rightarrow \mathcal{O}R}(\ketbra{y}{y}_I\otimes \rho_{E_SM_{\bar{S}}})U_{IE_SM_{\bar{S}}\rightarrow \mathcal{O}R}^\dagger = \ketbra{f_{S|(\alpha,\beta)}(y)}{f_{S|(\alpha,\beta)}(y)}_{\mathcal{O}}\otimes \sigma_R
    \end{align}
    for some state $\sigma_R$. We can CNOT the value of $f_{S|(\alpha,\beta)}(y)$ into a separate register, then invert the unitary to get what we started with.
    We then pick a new value of $y$ and repeat. 
    In this way we can determine the truth table of $f_{S|(\alpha,\beta)}$. 
    Since there is just one $\alpha$ that leads to $f_{S|(\alpha,\beta)}$ by assumption, this determines $\alpha$ as needed. 
    Since the $\alpha$ can be perfectly distinguished, we have $F(\rho_{E_SM_{\bar{S}}}(\alpha,\beta), \rho_{E_SM_{\bar{S}}}(\alpha',\beta))=\delta_{\alpha,\alpha'}$ as claimed. 

    Now, we apply \cref{lemma:subsystemsizebound}. 
    In our setting, $d_X$ is the total number of choices of $\alpha$, so $\log d_X$ is $g^*_S(f)$. 
    Thus we obtain
    \begin{align}
        \frac{1}{2}(1-3\delta/2)g^*_S(f) - h(\delta/2) \leq n_{E_S}.
    \end{align}
    We find that the entanglement held by $S$ must consist of at least $g^*_S(f)/2$ qubits, up to a multiplicative factor which is close to 1, and a small additive term. 
    Then we consider a partition $V=(V_1,...,V_m)$, use the above for each $V_i$, sum over the $V_i$, and optimize over partitions to get a lower bound in terms of $G^*(f)$. 
    Doing so leads to the claimed lower bound.
\end{proof}

We also remark that if $\PSQM^*$ protocols can be amplified, then the above lower bound technique can be adapted to apply to imperfectly correct protocols. 
To see why, suppose the $\PSQM^*$ protocol can be amplified in the sense that, with a factor of $\ell$ overhead in resources, we have $\epsilon\rightarrow 2^{-\ell}\epsilon$, and further that the security parameter is not increased by the amplification procedure. 
Then, we can consider a similar construction as used above, now applied to the amplified protocol. 
A complication arises when the referee wishes to determine $\alpha$ from $E_SM_{\bar{S}}$. 
To do so, they compute $f_{S|(\alpha,\beta)}(y)$ for a given value of $y$, measure the output qubit and store the value of $f_{S|(\alpha,\beta)}(y)$, then reverse the procedure to try and return to the initial state on $E_SM_{\bar{S}}$. 
With perfect correctness the measurement returns $f_{S|(\alpha,\beta)}(y)$ with probability 1 and doesn't change the state, so that the initial state on $E_SM_{\bar{S}}$ is reproduced perfectly. 
Repeating, the referee learns all the values $\{f_{S|(\alpha,\beta)}(y)\}_y$ correctly with probability 1. 
To ensure this works with probability of order 1 when the protocol is imperfectly correct, we need that each measurement fails to return $f_{S|(\alpha,\beta)}(y)$ with probability at most $2^{-|S|}$. 
Thus in the amplification step we should choose $k=\Theta(|S|)$. 
This means our lower bound is weakened by a factor of $|S|$, becoming,
\begin{align}
    n_{E_S} \gtrsim g_S^*(f)/|S|.
\end{align}
Thus on the total dimension of the entangled state we obtain a lower bound of $G^*(f)/\max_i |S_i|$ where the $\{S_i\}_i$ are the subsets in the optimizing partition used in computing $G^*(f)$. 

Many functions have large Ne\v{c}iporuk measure, and explicit functions can be identified with large Ne\v{c}iporuk measure as well. 
A review of this can be found in \cite{ball2022note}. 
To summarize briefly, we have that all functions have Ne\v{c}iporuk measure bounded by
\begin{align}
    G^*(f) \leq \frac{k^2n}{\log (kn)}.
\end{align}
Further, random functions nearly saturate this bound with high probability. 
Considering explicit functions, the $(n,k)$ set disjointness function ($DISJ_{n,k}$) and first-bit indirect storage access ($FISA_{n,k}$) function\footnote{See \cite{ball2022note}, section 4 for definitions of these functions.} have 
\begin{align}
    G^*(f)= \Omega\left(\frac{n^2k}{\log (kn)}\right)
\end{align}
which saturates the upper bound on $G^*$.
For these functions we obtain quadratic lower bounds on the entanglement cost of perfectly correct PSQM. 

\section{Rank lower bound}

We will give a lower bound strategy based on the rank of the communication matrix for the function $f(x,y)$.
This technique will only apply to the case where the PSM protocol is perfectly secure, but correctness errors are allowed.
This is similar to a technique given in \cite{asadi2024rank}, which deals with the related (quantum) CDS primitive \cite{GERTNER2000592,allerstorfer2024relating,asadi2025conditional}, but for quantum PSM the technique gives a stronger bound. 
In this section we deal with $2$-player PSM, but note that for $k$-players we can apply the technique to any partition of the players into two subsets. 

\subsection{Lower bound}

In this section we prove the rank lower bound on perfectly private PSM. 

\begin{theorem}\label{thm:ranklowerbound}
    Consider an $\epsilon$-correct, perfectly secure 2 party $PSQM^*$ protocol for the function $f:\{0,1\}^n\times \{0,1\}^n\rightarrow \{0,1\}$.
    Then
    \begin{align}
        \pp\overline{\PSQM}^*(f) \geq \frac{1}{4}\log \rank f - \frac{1}{4}
    \end{align}
    where $\rank(f)$ denotes the rank of the communication matrix for $f$, computed over the complex numbers. 
\end{theorem}

\begin{proof}
Consider a PSQM$^*$ protocol. 
Let the density matrix describing the message systems $M_AM_B=M$ be denoted $\rho_M(x,y)$. 
For a perfectly secure PSQM protocol, there are just two possible density matrices that can be realized for different values of $(x,y)$; a 0 density matrix or a 1 density matrix, 
\begin{align}
    \rho_M(x,y)=\rho_0 \quad \text{if} \,\,\,\, f(x,y)=0,\nonumber \\
    \rho_M(x,y)=\rho_1 \quad \text{if}\,\,\,\, f(x,y)=1.
\end{align}
This follows because if $\rho(x,y)\neq \rho(x',y')$, the referee has some probability of distinguishing between $(x,y)$ and $(x',y')$, which shouldn't be the case if $f(x,y)=f(x',y')$. 

This lets us express $f(x,y)$ using the matrix valued function $\rho(x,y)$, in particular define
\begin{align}
    \alpha = \tr((\rho_1-\rho_0)^2) = \Vert \rho_1-\rho_0\Vert_2^2.
\end{align}
This is non-zero, because by correctness it must be the case that $\rho_0\neq \rho_1$.
Then we can notice
\begin{align}\label{eq:f-as-rho}
    f(x,y) = \frac{1}{\alpha} \tr(\rho(x,y)-\rho_0)^2.
\end{align}
To see why the above holds, note that when $f(x,y)=1$ we have $\rho(x,y)=\rho_1$, so the right hand side is 1 as needed. 
Meanwhile, if $f(x,y)=0$ then $\rho(x,y)=\rho_0$ so the above is 0, as needed. 

Next, we use an expression for $\rho(x,y)$ that is constrained by the amount of entanglement used in the protocol. 
Suppose that the resource state used in the protocol is $\Psi_{L'R}$, and purify this to
\begin{align}
    \ket{\Psi}_{LR}=\sum_{i=1}^{r} \sqrt{\lambda_i}\ket{i}_L\ket{i}_R. \label{eq:ent-purification}
\end{align}
We will lower bound the log-dimension of the purified state, which never need be larger than twice the unpurified log-dimension. 
Thus beginning with a mixed state and purifying will loosen our bound by at most a factor of $1/2$. 

The mid-protocol density matrix $\rho_{M_AM_B}(x,y)$ is produced by Alice and Bob acting separately on each end of the entangled state, so is of the form
\begin{align}
    \rho_{M_AM_B}(x,y) &=\sum_{i,j=1}^{r} \sqrt{\lambda_i\lambda_j} \mathcal{N}_{L\rightarrow M_A}^x(\ketbra{i}{j})\otimes \mathcal{N}_{R\rightarrow M_B}^y(\ketbra{i}{j}) \nonumber \\
    &= \sum_{i,j=1}^{r} A^{i,j}_{M_A}(x)\otimes B^{i,j}_{M_B}(y)\nonumber \\ 
    &= \sum_{I=1}^{r^2} A^{I}_{M_A}(x)\otimes B^{I}_{M_B}(y).
\end{align}
The second line defines the matrices $A^{i,j}_{M_A}(x), B^{i,j}_{M_B}(y)$. 
Pick any input $(x_0,y_0)$ with $f(x_0,y_0)=0$. 
Then
\begin{align}
    \rho_0 = \sum_{I=1}^{r^2} A^I_{M_A}(x_0)\otimes B^I_{M_B}(y_0)
\end{align}
and we can write $f(x,y)$ as
\begin{align}
    f(x,y) &= \frac{1}{\alpha} \tr(\rho(x,y)-\rho_0)^2 \nonumber \\
            &= \frac{1}{\alpha}\tr\left((\rho(x,y))^2 - 2\rho(x,y) \rho_0 + \rho_0^2 \right)\nonumber \\
           &= \frac{1}{\alpha} \tr \left( \sum_{I,J=1}^{r^2} A^I(x)A^J(x)\otimes B^I(y)B^J(y) - 2\sum_{I,J=1}^{r^2} A^I(x)A^J(x_0)\otimes B^I(y)B^J(y_0)\right. \nonumber \\
           &\qquad \qquad \qquad + \left.\sum_{I,J=1}^{r^2} A^I(x_0)A^J(x_0)\otimes B^I(y_0)B^J(y_0)\right) \nonumber \\
           &= \sum_{K=1}^{2r^4} f_{K}(x)f'_{K}(y)
\end{align}
where the last line defines two functions $f_K(x)$ and $f_K'(y)$. Therefore, we have that the rank of $f$ over the complex numbers is at most 
\begin{align}
    \rank(f) \leq 2r^4. \label{eq:rank-bounds-ent}
\end{align}
Since $r$ is the Schmidt rank of the purified state, $\log r$ lower bounds the dimension of $R$ and $L$, which we want to translate to a bound on the dimension of $\rho_{L'R}$. 
We know that
\begin{align}
    \log d_{LR} \leq 2 \log d_{L'R} \label{eq:ent-dimension-blowup}
\end{align}
and that $\log r \leq \log d_R, \log d_L$.
Combining these statements we have
\begin{align}
    \log r \leq \frac{1}{2}\log d_{LR} \leq \log d_{L'R} = \pp\overline{\PSQM}^*(f).
\end{align}
Finally, by \cref{eq:rank-bounds-ent} we have that $\log(r) \geq \frac{1}{4} \log\rank(f) - \frac{1}{4}$. Combining this with the previous equation we get
\begin{align}
    \boxed{\pp\overline{\PSQM}^*(f) \geq \frac{1}{4}\log \rank f - \frac{1}{4}.}
\end{align}
Here $\pp\overline{\PSQM}^*(f)$ is the minimal number of qubits of shared resource state needed to execute PSQM with perfect security. 
\end{proof}

Note that the rank lower bound does not apply to communication, but only to correlation cost.  
We could ask if the communication is also lower bounded by the rank, but considering the equality function shows such a bound cannot hold. 
To see this, consider that equality is full rank, so the log-rank is $n$. 
Using linear randomness in the resource, Alice and Bob can take $x\oplus r$ and $y\oplus r$ and then send a constant length hash of those strings. 
This is perfectly secure and approximately correct, so we can compute equality in the PSQM$^*$ model with perfect privacy and log communication, even while the log-rank is linear. 
Thus our bound cannot also apply to communication. 

\subsection{Comparison to related bounds}

We should also compare the rank lower bound obtained here to similar lower bounds. 
We can notice first that the communication cost $\PSQM^*(f)$ is lower bounded by two-way quantum communication complexity of $f$.
If we assume perfect correctness, this is lower bounded by the log-rank \cite{buhrman2001communication}. 
If we had assumed perfect correctness for our PSQM protocol then our lower bound would be immediate. 
Our result shows that perfect privacy (and relaxed correctness) also suffices to obtain a rank bound. 

Note that our bound also applies to classical PSM, since classical PSM is lower bounded by quantum PSM. 
As well, our lower bound uses privacy --- indeed our bound is on the shared randomness, which can be zero when there is no privacy requirement. 

Another point of comparison are the similar lower bounds on quantum CDS obtained in \cite{asadi2024rank}. 
There, the log Schmidt rank in perfectly private quantum CDS, denoted $\pp\CDQS$, is lower bounded according to
\begin{align}
    \pp\CDQS^*(f) &\geq \frac{1}{4}\log (\text{nrank}(f)).
\end{align}
The non-deterministic rank is the minimal rank of any matrix over the complex numbers which has the same zeros as the communication matrix of $f$. 
In general this may be smaller than the usual notion of rank. 
Because CDS lower bounds PSM \cite{allerstorfer2024relating}, we also obtain the same lower bound on perfectly private PSM, but since the nrank may be much smaller than the rank, our bound on PSM is stronger than the one inherited from CDS.

\section{Upper bound from \texorpdfstring{$T$}{T}-depth and quantum circuits}

In this section we prove new upper bounds on the $\PSM_k^*$ model. 
Our main strategy is to adapt techniques from non-local quantum computation \cite{speelman2015instantaneous}, which involves two separated players, to the $k$-player setting. 
This gives an upper bound on $\PSM^*_k(f)$ related to the $T$-depth of any unitary computing $f$. 
With some further modifications to this technique, we show it can be used to realize an upper bound on $\PSM^*_k$ based on the size and depth of a quantum circuit computing $f$. 

\subsection{Computational models}\label{sec:computationalmodels}

Recall that the $n$-qubit Pauli group $\mathcal{P}_n$ consists of $n$-fold tensor products of the four Pauli operators, $\{I,X,Y,Z\}$, with an added overall phase of $\pm 1$ or $\pm i$. 
A Clifford unitary $C$ is a unitary such that for any $P\in \mathcal{P}_n$, there exists another $P'\in \mathcal{P}_n$ such that $CPC^\dagger=P'$.

Clifford unitaries are a subgroup of all possible unitaries. 
To generate the full unitary group, it suffices to consider the $T$-gate, 
\begin{align}
    T=\begin{pmatrix} 1 & 0 \\ 0 & e^{i\pi/4}\end{pmatrix}
\end{align}
along with the Cliffords. 
$T$-gates are more difficult to apply in standard quantum computing architectures, and consequently understanding the number of $T$-gates required for a given computation has been a topic of interest. 

We will be interested in decompositions of unitaries into Clifford layers interlaid with $T$-gate layers, 
\begin{align}
    U=C_{d}\bar{T}_dC_{d-1}...C_1\bar{T}_1C_0.
\end{align}
The operators $\bar{T}_i$ consist of parallel $T$-gates on a subset $S_i$ of the qubits, and identity on the remaining qubits. 
The operators $C_i$ are Clifford. 
We call $d_T$ the $T$-depth of the circuit. 

We are also interested in restricted variants of the above decomposition into Clifford$+T$ layers. 
To specify these, we define a unitary $V$ to have a backward light-cone of size at most $\ell$ if, taking a single qubit operator $\mathcal{O}_i$, we have that
\begin{align}
    V^\dagger \mathcal{O}_i V = \tilde{\mathcal{O}}_i
\end{align}
has support on at most $\ell$ qubits. 
Note that we do not require any notion of geometric locality in these operators; we only consider the size of the non-trivial support. 
Returning to Clifford$+T$ decompositions, we say that a Clifford$+T$ decomposition is $\ell$-local if every Clifford $C_i$ has backward light-cones of size at most $\ell$. 

\subsection{Speelman's instantaneous computation technique}

To prove our upper bounds in this section, we borrow techniques developed in the context of non-local quantum computation (NLQC) \cite{speelman2015instantaneous}. 
In both NLQC and in the context of PSM protocols, it is useful to be able to implement the following transformation. 

\begin{definition}
    An \emph{instantaneous computation} of a unitary $U_{AB_1...B_{k-1}}$ is the following transformation. 
    $k$ players, whom we call Alice and Bob$_1$, ..., Bob$_{k-1}$ each hold one of the systems $A$, $B_1$, ..., $B_{k-1}$.
    The choice of unitary $U_{AB_1...B_{k-1}}$ is known to all players.
    We say that Alice and Bob instantaneously implement $U_{AB_1...B_{k-1}}$ if they implement the transformation
    \begin{align}
        \ket{\psi}_{AB'}\rightarrow P_{AB'}[m_a,m_b^1,...m_{b}^{k-1}]U_{AB'}\ket{\psi}_{AB'}
    \end{align}
    where $AB'=AB_1...B_{k-1}$ is held by Alice at the end of the protocol, $P[m_a,m_b]$ is a Pauli string determined by $m_a,m_b$, Alice holds $m_a$, and Bob$_i$ holds $m_b^i$. 
\end{definition}

The entanglement cost of instantaneous computations can be upper bounded in terms of the complexity of the unitary $U_{AB'}$. 
Consider a decomposition of $U_{AB'}$ into Clifford layers and layers of $T$-gates, 
\begin{align}
    U=C_{d_{T}} \bar{T} C_{d_{T}-1} \bar{T} C_{d_{T}-2}... C_{1} \bar{T}C_0.
\end{align}
If we allow arbitrary Cliffords $C_i$ at each layer, and consider only 2 players Alice and Bob$_1$, \cite{speelman2015instantaneous} shows that $U$ can be computed instantaneously using $O((K_1n)^{d_T})$ shared EPR pairs where $d_T$ is the $T$-depth, $n$ is the number of qubits of input, and $K_1$ is a constant. 
We will need an extension of this result. 

Specifically, we consider the case where 1) We allow $k\geq 2$ players. 2) The Clifford layers are further restricted, to only allow Cliffords with backward light cones of size $\ell$. 
This is equivalent to the statement that the $i$th output qubit of the Clifford is only influenced by at most $\ell$ input qubits. 

\begin{theorem}\label{thm:unitary-restricted-clifford-t-depth}
Suppose that a unitary $U_{AB_1...B_{k-1}}$ on $\ntot$ qubits is of the form 
\begin{align}
    U=C_{d_{T}}^\ell \bar{T} C^\ell_{d_{T}-1} \bar{T} C_{d_{T}-2}^\ell... C_{1}^\ell \bar{T}C_0^\ell.
\end{align}
where each $C^\ell_i$ is a Clifford unitary with backward light cones of size at most $\ell$. 
Then $U_{AB}$ can be computed instantaneously using $E(U)$ shared EPR pairs, where
\begin{align}
    E(U)\leq O((K\ell)^{d_T-1}\cdot \ntot \cdot \min\{k,\ell\}).
\end{align}
Here $K$ is a constant independent of the choice of unitary. 
\end{theorem}
We prove this theorem in Appendix \ref{sec:kplayerSpeelman}.

Next we proceed to apply this result to $\PSM^*_k$ upper bounds. 

\subsection{\texorpdfstring{$T$}{T}-depth and circuit depth \texorpdfstring{$\PSM_k^*$}{PSMk*} upper bounds}

We first of all show the following $\PSM^*$ upper bound from the $T$-depth. 

\begin{theorem}\label{thm:restrictedCliffordTdepth}
Suppose that a Boolean function $f:X_1\times...\times X_k\rightarrow \{0,1\}$, $X_i=\{0,1\}^n$ can be computed from the inputs $\ket{x_1,...,x_k}$ and advice state $\ket{\psi}$ $\epsilon$-correctly by a unitary of the form
\begin{align}
    U=C_{d_{T}}^\ell \bar{T} C^\ell_{d_{T}-1} \bar{T} C_{d_{T}-2}^\ell... C_{1}^\ell \bar{T}C_0^\ell.
\end{align}
with Clifford layers $C_{i}^\ell$ restricted to Cliffords with light cones of size at most $\ell$. 
The total number of qubits of input plus advice is denoted $\ntot$.
Then there is an $\epsilon$-correct, $\delta=2\epsilon$ secure $\PSM^*$ protocol for $f$ with communication cost
\begin{align}
    \PSM_k^*(f)\leq O((K\ell)^{d_T-1}\cdot \ntot \cdot \min\{k,\ell\})
\end{align}
and which uses a number of EPR pairs satisfying the same upper bound.
\end{theorem}

\begin{proof}
Let the unitary computing $f(x)$ be computed by applying a unitary $U_{AB_1...B_{k-1}}$ to the inputs and advice state followed by a computational basis measurement on the first output qubit. 
By assumption the circuit computes $f$ with probability $1-\epsilon$. 

The $\PSM^*$ protocol is as follows. 
The players execute the instantaneous computation for $U$, which uses entanglement upper bounded as in \cref{thm:unitary-restricted-clifford-t-depth}. 
At this point, Alice holds the output of $U$, which has been corrupted by a Pauli string $P[m_a, m_{b_1},...,m_{b_{k-1}}]$ with the $b_i$ held by Bob$_i$. 
Alice proceeds to measure the first qubit, obtaining outcome $s$. 
All of the players then send the measurement outcomes $m_a,m_{b_i}$ produced during the execution of the instantaneous computation protocol. 

To see that this protocol is $\epsilon$ correct, notice that the Pauli operator acting on the measured qubit will not change the measurement outcome if it is an $I$ or $Z$ operator, and will invert the outcome if it is an $X$ or $Y$ operator. 
Thus if the referee receives this final measurement outcome along with the strings $m_a, m_{b_1},...m_{b_{k-1}}$, the referee can determine which Pauli acted on the measured qubit and inverts the result if it is $X$ or $Y$, then the resulting bit will be distributed just as if it were the outcome from the original circuit with no Pauli corrections, so will equal $f(x,y)$ with probability at least $1-\epsilon$. 

Next we consider security. 
In the instantaneous computation protocol, all of the bits sent to the referee were from Bell basis measurements, call them $r=(r_1,...,r_k)$, except one bit $s$, which came from Alice measuring the final output qubit. 
The Bell basis measurement outcomes $r$ are distributed as a uniformly random bit-string in $\{0,1\}^{|r|}$. 
To design a simulator, consider that the message is of the form
\begin{align}
    \rho_M(x,y)&=\frac{1}{2^{|r|}}\sum_r X^{p(r)} \sigma(x,y)X^{p(r)}\otimes \ketbra{r}{r}, \nonumber \\
    \sigma(x,y) &= \alpha(x,y) \ketbra{f}{f} + (1-\alpha(x,y)) \ketbra{f\oplus 1}{f\oplus 1}.
\end{align}
Here $p(r)$ is a function which determines if there is a Pauli $X$ correction on the measured qubit. 
The probabilities $\alpha(x,y)$ can in general leak information about $(x,y)$, but we have that $\alpha(x,y) \geq 1-\epsilon$ for all $(x,y)$ which will ensure this leaked information is small. 
In particular we define the simulator distribution to be
\begin{align}
    \text{Sim}(f) = \frac{1}{2^{|r|}}\sum_r X^{p(r)} \ketbra{f}{f}  X^{p(r)}\otimes \ketbra{r}{r}.
\end{align}
Then to check security, we just need to calculate the trace distance between the message distribution and the simulator distribution, 
\begin{align}
    \Vert \rho_M(x,y) - \text{Sim}_M(f) \Vert_1 &= \vabs{\frac{1}{2^{|r|}}\sum_r X^{p(r)}(\sigma(x,y)-\ketbra{f}{f})X^{p(r)}\otimes \ketbra{r}{r} }_1 \nonumber \\
    &= \frac{1}{2^{|r|}}\sum_r \Vert \sigma(x,y)-\ketbra{f}{f} \Vert_1 \nonumber \\
    &= \frac{1}{2^{|r|}} \sum_r \|(\alpha(x,y)-1)\ketbra{f}{f} + (1-\alpha(x,y))\ketbra{f\oplus 1}{f\oplus 1} \|_1\nonumber \\
    &\leq \frac{1}{2^{|r|}} \sum_r 2 |1-\alpha(x,y)| \nonumber \\
    &\leq 2\epsilon
\end{align}
so that the protocol is $\delta=2\epsilon$ secure, as claimed. 
\end{proof}

As a special case of this result, it follows that low-depth quantum circuits, with gates of constant arity, provide a good upper bound on $\PSM^*$ complexity. 

\begin{corollary}
Suppose that a Boolean function $f:X_1\times...\times X_k\rightarrow \{0,1\}$, $X_i=\{0,1\}^n$ can be computed from the inputs $\ket{x_1,...,x_k}$ by a depth $d_f$, size $s$ quantum circuit (composed of gates acting on $O(1)$ qubits). 
Then, there is an $\eps$-correct, $\delta = 2\eps$-secure $\PSM^*$ protocol for $f$ with communication cost
\begin{align}
    \PSM_k^*(f) \leq (kn +s) \cdot \log^{O( d_f)}(s/\eps)
\end{align}
and which uses a number of EPR pairs satisfying the same upper bound.

\end{corollary}
\begin{proof}
First we convert the gates from the given quantum circuit to the Clifford + T gate set that consists only of one and two-qubit gates, in particular our gate set consists of: Hadamard, Phase, T, and CNOT. Using Solovay-Kitaev \cite{kitaev1997quantum,nielsen2010quantum} we could do this with a depth overhead factor of $\log(s/\eps)$ where $s$ is the size of the circuit. However, we will do better than this by making use of an additional \emph{catalytic} \cite{kim2025catalytic} advice state. The catalytic advice state is provided as additional input to the circuit (does not depend on the input) and after the computation it must be returned to its original form so can be reused. Surprisingly, Kim showed that any Pauli $Z$-rotation can be $\eta$-approximated in depth 3 if a certain catalyst state is available, which has size $O(\log^2(1/\eta))$ \cite{kim2025catalytic}. Kim and Laakkonen later improved this to be depth 1 with size $O(\log(1/\eta))$ advice \cite{kim2025catalytic,kim2025any}\footnote{Kim notes in v2 of \cite{kim2025catalytic} that after the release of the first version of the paper, Craig Gidney observed that it could be done in T-depth 2 \cite{gidney_x_1936285631359197210}.}. We will use this construction to convert our circuit to a Clifford + T circuit with low T-depth. 

For each gate in the original circuit we do the following. 
Let $b = O(1)$ be the bound on the number of qubits of fan-in the circuit's gates have. 
It's well-known that we can implement any unitary of dimension $2^b$ exactly with a circuit of size $4^{b}$ using only arbitrary single qubit rotations about $Z$ and $Y$, and CNOT gates \cite{mottonen2004quantum}. 
Note that the Y-rotation gate is just a Z-rotation gate conjugated by the single-qubit Clifford gate, $SH$. 
Thus, we can rewrite the gate with $O(4^b)$ gates consisting only of Clifford and $Z$-rotations. 
Using Kim's construction \cite{kim2025catalytic,kim2025any} we then approximate each of the $O(4^b)$ $Z$-rotation gates to within approximation error $\eta$, which only requires $T$-depth 1 for each such $Z$-rotation. 
We implement each of the gates sequentially so that we can reuse the catalytic state. Now the overall T-depth for implementing this gate is $O(4^b)$, uses a catalytic state on $O( \log(1/\eta))$ qubits, and has approximation error $\eta \cdot O(4^b)$.

We do the above for each of the gates. 
Let $s$ be the total number of gates in the original circuit. 
So our constructed Clifford + T circuit then has T-depth $d_T := O(d_f \cdot 4^b)$, approximation error $\eps = O(s \eta 4^b)$ and requires a catalytic advice state on $O(s \cdot \log(1/\eta))$ qubits. 

Furthermore, note that each of the Clifford components act on at most $O(b + \log(1/\eta))$ qubits: the original $b$ input qubits and the $\log(1/\eta)$ catalytic advice qubits. So each of the Clifford layers has lightcones bounded by $\ell := O(b +  \log(1/\eta))$.

We set $\eta = \Theta(\eps s^{-1} 4^{-b})$ such that the total approximation error is $\eps$, the catalytic state is on $a:= O(s \cdot \log(s 4^b / \eps)) = O(s b \log(s/\eps))$ qubits, and the lightcones are bounded by $\ell = O(b\log(s/\eps))$.

Let $U$ be the unitary on $\ntot' := (\ntot + a) = nk + O(s b \log(s/\eps))$ qubits that implements this circuit, so it takes as input the input to the original circuit ($nk$-qubits) in addition to the ($a$-qubit) catalytic advice state, and it computes $f$ to within approximation error~$\eps$. Using that $\ntot = k n$ we get the claimed bound.

We can implement this unitary instantaneously using \cref{thm:unitary-restricted-clifford-t-depth}. Now, Alice will receive $n$ bits of input in addition to the $a$ qubits of advice.  Applying \cref{thm:unitary-restricted-clifford-t-depth} with $\ntot'$ as the total number of qubits, we see that there is an instantaneous protocol for implementing $U$ with the number of EPR pairs used at most
\begin{align}
    E(U)&\leq O((K\ell)^{d_T-1}\cdot \ntot' \cdot \min\{k,\ell\})\\
    &= (b\log(s/\eps))^{O(d_f 4^b)} \cdot (\ntot + O(sb \log (s/\eps))) \cdot \min\{k, b\log(s/\eps) \} \\
    &\leq (b\log(s/\eps))^{O(d_f 4^b)} \cdot (\ntot+s)\\
    &\leq (\ntot+s) \cdot \log^{c\cdot d_f}(s/\eps)
\end{align}
for some constant $c>0$. In the last line we used that $b = O(1)$.

Now for the $\PSM_k^*$ protocol, when each player is given the classical $n$-bit input, Alice can prepare the catalytic advice state herself, the players can apply the above protocol.
Furthermore, following the same argument as in the proof of \cref{thm:restrictedCliffordTdepth}, this protocol is $\eps$-correct and $2\eps$-secure, and both communication cost and number of EPR pairs at most $(\ntot+s) \cdot \log^{c\cdot d_f}(s/\eps)$. 

\end{proof}

\section{Upper bound from the Fourier 1 norm}\label{sec:fourierupperbound}

In this section we give an upper bound on classical PSM complexity based on the one-norm of the Fourier transformation of $f$. 
This upper bound appears to be new to the classical literature. 
As we discuss, the Fourier one-norm captures how the function $f$ is related to parity functions. 

\subsection{Fourier transform for Boolean functions and the Fourier one-norm}\label{sec:fourier}

We briefly introduce the Fourier transform of a Boolean function. 
Let $x,S$ be $n$ bit strings. 
We think of $S$ as labelling a subset of the bits, where the $i$th bit of $S$ is $1$ iff bit $i$ is included in the subset. 
Then define the parity function
\begin{align}
    \chi_S(x)=(-1)^{S\cdot x} = (-1)^{\sum_{i\in S}x_i}
\end{align}
which is $-1$ if the parity of the bits $x_i$ in the subset $S$ is odd, and $+1$ otherwise. 
We study a Fourier transform of $f$, which can also be understood as expressing $f(x)$ as a sum over parity functions.
Specifically, the Fourier transform $\hat{f}(S)$ is defined such that
\begin{align}
    f(x)=\sum_S \hat{f}(S)\chi_S(x)
\end{align}
Alternatively, we can use that 
\begin{align}
    \frac{1}{2^n}\sum_x \chi_S(x)\chi_T(x)= \delta_{S,T}
\end{align}
to express $\hat{f}(S)$ in terms of $f(x)$, 
\begin{align}
    \hat{f}(S)=\frac{1}{2^n}\sum_x f(x) \chi_S(x).
\end{align}

Our upper bound will be expressed in terms of the Fourier 1 norm, which given a function $f$ is defined from the Fourier coefficients $\hat{f}(S)$ according to
\begin{align}
    \Vert \hat{f} \Vert_1 = \sum_{S} |\hat{f}(S)|
\end{align}
The Fourier 1 norm captures, in a weighted sense, how many Fourier coefficients are needed to represent $f(x)$. 

\subsection{Upper Bound}

\begin{theorem}\label{thm:psm_fourier_norm}
    For any function $f:X\times Y\rightarrow \{0,1\}$, and any $\delta>0$, the function $f$ can be computed $\delta$-correctly and perfectly securely in the PSM model using communication and randomness both of $O(\|\hat{f} \|_1^2 \ln(2/\delta) )$. 
\end{theorem}

\begin{proof} The starting point for the proof is the observation that $f(x,y)$ can be expressed as the expected value of a random variable from a certain fixed distribution depending on the function $f$. In more detail, define a probability distribution $p$ from the Fourier coefficients of $f$, 
\begin{align}
    p_S=\frac{|\hat{f}(S)|}{\Vert \hat{f} \Vert_1}.
\end{align}

Observe that $f(x,y)$ can be expressed in terms of the expectation of the random variable $ \text{sign}(\hat{f}(S))\cdot \chi_S(x,y)$ under this distribution. 
In more detail, 
\begin{align}
\begin{split} f(x,y)&=\sum_S \hat{f}(S)\chi_S(x,y)\\
&=\|\hat{f}\|_1\cdot \sum_{S} p_S\cdot   \text{sign}(\hat{f}(S))\cdot \chi_S(x,y)\\ &=\|\hat{f}\|_1\cdot \mathbb{E}_{S\sim p}[ \text{sign}(\hat{f}(S))\cdot \chi_S(x,y)].\end{split}
\end{align}
Next, we note that $\chi_{S}(x,y)=\chi_S(x)\chi_S(y)$, as can be checked from its definition.
Thus,  
\begin{align}\label{eq:bias}
   \frac{f(x,y)}{\|\hat{f}\|_1}=\mathbb{E}_{S\sim p}[\text{sign}(\hat{f}(S))\cdot \chi_S(x)\cdot \chi_S(y)].
\end{align}
This naturally motivates a strategy for Alice and Bob as follows. 
Firstly, using shared private randomness, they sample a set $S\sim p$ according to the aforementioned distribution $p$, secondly, they sample a uniformly random bit $b\in \{-1,1\}$. 
Then, Alice sends $\chi_S(x)\cdot b$ to the referee and Bob sends $\chi_S(y)\cdot b\cdot \text{sign}(\hat{f}(S))$. 
We then have the referee multiply these two quantities to get $\text{sign}(\hat{f}(S))\cdot \chi_S(x,y)$.

Firstly, we observe that the expected output of the referee is precisely $f(x,y)/\|\hat{f}\|_1$ from \cref{eq:bias}. 
Furthermore, we will see that this protocol is perfectly private. But first, we would like to amplify the success probability to get perfect correctness. 
To determine whether $f(x,y)$ is $1$ or $-1$ with probability $1-\delta$, the idea is to have the players repeat this protocol $m=10 \|\hat{f}\|_1^2\ln(2/\delta)$ times in parallel (using independent randomness) and have the referee output $1$ or $-1$ if the sum of all his outputs is positive or negative respectively. 
We now prove that this described protocol is correct. 
To do so, observe that expected sum of the outputs is $m\cdot f(x,y)/\|\hat{f}\|_1$. 
As each output is $\{\pm 1\}$-valued, by Hoeffding's inequality, the probability that the sum differs from its expectation by more than $t$ is at most
$2\exp\left(-t^2/2m\right)$, which is at most $\delta$ for $t=\sqrt{2m \ln(2/\delta)}$. 
Thus, with all but $\delta$ probability, the sum of all outputs is roughly $m\cdot f(x,y)/\|\hat{f}\|_1\pm \sqrt{2m\ln(2/\delta)}$ which is roughly $10\|\hat{f}\|_1\ln (2/\delta)\cdot f(x,y)\pm 5\|\hat{f}\|_1\ln(2/\delta)$ by our choice of $m$. 
As the first term dominates, it is clear that this quantity is positive if and only if $f(x,y)=1$. 
This completes the proof of correctness of the amplified protocol. 

Perfect privacy of the original (not repeated) protocol, implies perfect privacy of the repeated protocol, so all that remains is to show that the original protocol is perfectly private. 
Recall that the referee receives a bit $b_1\in \{-1,1\}$ from Alice and a bit $b_2\in\{-1,1\}$ from Bob, where $b_1=b\cdot \chi_S(x)$ for a uniformly random bit $b\in\{-1,1\}$ and $b_2=b\cdot \chi_S(y)\cdot \text{sign}(\hat{f}(S))$. 
The simulator for this protocol is as follows: Given $f(x,y)$, the simulator samples a uniformly random bit $b\in \{-1,1\}$ and an independent random bit $c\in \{-1,1\}$ with bias given by $f(x,y)/\|\hat{f}\|_1$ and outputs $b_1':=b$ for Alice's message and $b_2':=b\cdot c$ for Bob's message. 
We will show that the resulting distribution on $(b_1',b_2')$ is identical to that of $(b_1,b_2)$ in the original protocol. 
We will instead argue that the distribution of $(b_1',b_1'\cdot b_2')$ is identical to that of $(b_1,b_1\cdot b_2)$ and since $b_1,b_1',b_2,b_2'\in\{-1,1\}$, this suffices. 
Firstly, observe that these are product distributions as the random variables $b_1\cdot b_2$ and $b_1'\cdot b_2'$ are independent of $b_1$ and $b_1'$ respectively. 
Thus, it suffices to argue that the marginal distributions are identical. 
Firstly, it is clear that $b_1$ and $b_1'$ are uniformly random bits. Secondly, observe that $b_1'\cdot b_2'$ is by definition a random bit in $\{-1,1\}$ with bias $f(x,y)/\|\hat{f}\|_1$. 
In the original protocol, we observe that $b_1\cdot b_2$ is a bit obtained by sampling $S\sim p$ and returning $\chi_S(x,y)\cdot \text{sign}(\hat{f}(S))$ -- this is simply a random bit whose expectation is $f(x,y)/\|\hat{f}\|_1$ by \cref{eq:bias}. 
This completes the proof.
\end{proof}

\appendix

\section{\texorpdfstring{$k$}{k}-player instantaneous computations based on the \texorpdfstring{$T$}{T}-depth}\label{sec:kplayerSpeelman}

In this appendix we give the proof of \cref{thm:unitary-restricted-clifford-t-depth}. 
Our techniques follow \cite{speelman2015instantaneous}, but we make some generalizations to achieve the needed bound. 

\subsection{The garden-hose model}

To give the proof, we first need to describe the garden-hose model and some associated instantaneous computation strategies. 
The garden-hose model \cite{buhrman2013garden} is most easily described in terms of the following setting. 
Alice and Bob are neighbours, and share a fence. 
Alice has an input string $x\in\{0,1\}^n$, while Bob has an input string $y\in\{0,1\}^n$. 
Alice has a tap, which she can turn on to produce a flow of water. 
Alice and Bob share a number of pipes which connect their yards, and they have hoses that they can use to connect pipes to one another, or the tap to a pipe. 
Alice and Bob wish to compute a Boolean function $f(x,y)$, with the outcome determined by where the water spills. 
Typically, the model is defined so that water spilling on Alice's side indicates $f(x,y)=0$, while water spilling on Bob's side indicates $f(x,y)=1$. 
The garden-hose model can also be formalized in terms of path connectivity in a particular form of graph, see \cite{buhrman2013garden}.
We won't need this formalization though, and stick to the more informal water-based description. 

We modify the definition of the garden-hose model somewhat and specify two pipes labelled ``$f=0$'' and ``$f=1$'', and require the water to spill on Alice's side out of the corresponding pipe.
In fact, these models are not too different, as the following lemma from \cite{klauck2014new} shows. 
\begin{lemma}
    Suppose there is a garden-hose protocol that computes $f(x,y)$ using $m$ pipes in the sense that water spills on Alice's side if $f(x,y)=0$, and on Bob's side if $f(x,y)=1$. 
    Then there is also a garden-hose protocol that computes $f(x,y)$ in the sense of water spilling on Alice's side from one of two designated pipes that uses at most $3m+1$ pipes. 
\end{lemma}
The model where water spills from designated pipes will generalize better to the $k$ player setting, so take the model with two designated pipes on Alice's side as the defining setting. 
We denote the number of pipes needed to compute $f$ in (this variant of) the garden-hose model by $GH(f)$.

In the quantum context, the garden-hose model appears as a description of concatenated teleportations in some settings. 
In particular, consider an unknown quantum state $\ket{\psi}$, which plays the role of the tap in the garden-hose description. 
Alice and Bob share a set of EPR pairs between them, which play the role of the pipes. 
Alice and Bob can then make Bell basis measurements, which act on either two ends of EPR pairs they hold in their own labs, or (in Alice's case) on the input state plus the end of one EPR pair.
To see why the water-flow analogy of the garden-hose model is relevant, consider that after the input state is measured with one EPR pair, the state has moved to the other end of the EPR pair, up to Pauli corrections. 
Each subsequent measurement moves the state to the other end of the measured EPR pair. 
To an observer with access to the measurement outcomes, it is as if the state is flowing along the path determined by the pipes in the garden-hose picture. 

We generalize the garden-hose model to allow $k\geq 2$ players, in which case we name the players Alice, Bob$_1$, Bob$_2$,..., Bob$_{k-1}$.
Each player receives a string of $n$ bits, with $x_0$ the string held by Alice and $x_i$, $k-1\geq i\geq 1$ held by Bob$_i$.
In this setting, pairs of players share pipes, and the actions of the players is as before: Alice may connect the tap to a pipe, and all players may connect open ends of pipes they have access to using hoses. 
Again the result of the computation is determined by considering two labelled pipes shared with Alice. 
This model can be instantiated as before in the quantum context with the tap replaced by an input state and the pipes replaced by EPR pairs. 
We define the minimal number of pipes used in a $k$-party garden-hose protocol by $GH_k(f)$. 
Note that we allow any fixed configuration of pipe connections among the players. 

We make use of the following lemma regarding the $k$-party garden-hose model.
The proof is nearly the same as is given in \cite{klauck2014new, speelman2015instantaneous} for the case of two players.
For completeness, we include the proof in our $k$-player setting.

\begin{lemma}\label{lemma:XORlemma}
The $k$-party garden-hose complexity satisfies the following:
\begin{align}
    GH_k\left(\bigoplus_i f_i\right) \leq 4 \sum_i GH_k(f_i)
\end{align}
\end{lemma}

\begin{proof}
Consider garden-hose protocols for each $f_i$, which we label $P_i$. 
We give a garden-hose protocol for $\oplus_i f_i$ by wiring copies of the $P_i$ together in an appropriate way. 
Concretely, we take 4 copies of $P_i$, and connect them as shown in \cref{fig:kpartyGH}. 
The gadget has four open hoses, which we wire together with further gadgets: we wire the $0$ output of the $f_i$ gadget to the $0$ input of the $f_{i+1}$ gadget, and the $1$ output of the $f_i$ gadget to the $1$ input of the $f_{i+1}$ protocol. 
By inspection, one can check that the gadget flips the parity of the input if $f_i=1$, and leaves the input unchanged if $f_i=0$.
To compute $\oplus_i f_i$ then, we connect the tap to the $0$ input of the $f_1$ gadget, and label the $0$ and $1$ outputs of the final $f_i$ as the $0$ and $1$ labelled output hoses of the protocol. 
After the water flows through gadgets for each $f_i$, we've computed $\oplus_i f_i$. 
\end{proof}

\begin{figure}
    \centering
    \begin{tikzpicture}

    \draw[thick] (0,0) -- (0,2.5) -- (5,2.5) -- (5,0) -- (0,0);
    \draw[thick,blue] (-1,3) to [out=-90,in=180] (0,2);
    \node[above] at (-1,3) {$0$ in};
    \draw[thick] (0,1) -- (-1,1);
    \draw[thick] (0,0.5) -- (-1,0.5);
    \node at (2.5,1.25) {GH protocol for $f_i$};

    \begin{scope}[shift={(9,0)}]
    \draw[thick] (0,0) -- (0,2.5) -- (5,2.5) -- (5,0) -- (0,0);
    \draw[thick,blue] (-1,3) to [out=-90,in=180] (0,2);
    \node[above] at (-1,3) {$1$ in};
    \draw[thick] (0,1) -- (-1,1);
    \draw[thick] (0,0.5) -- (-1,0.5);
    \node at (2.5,1.25) {GH protocol for $f_i$};
    \end{scope}

    \begin{scope}[shift={(0,-5)}]
    \draw[thick] (0,0) -- (0,2.5) -- (5,2.5) -- (5,0) -- (0,0);
    \draw[thick,blue] (-1,-0.5) to [out=90,in=180] (0,2);
    \node[below] at (-1,-0.5) {$0$ out};
    \draw[thick] (0,1) -- (-1,1);
    \draw[thick] (0,0.5) -- (-1,0.5);
    \node at (2.5,1.25) {GH protocol for $f_i$};
    \end{scope}

    \begin{scope}[shift={(9,-5)}]
    \draw[thick] (0,0) -- (0,2.5) -- (5,2.5) -- (5,0) -- (0,0);
    \draw[thick,blue] (-1,-0.5) to [out=90,in=180] (0,2);
    \node[below] at (-1,-0.5) {$1$ out};
    \draw[thick] (0,1) -- (-1,1);
    \draw[thick] (0,0.5) -- (-1,0.5);
    \node at (2.5,1.25) {GH protocol for $f_i$};
    \end{scope}

    \draw[thick] (-1,1) to [out=180,in=180] (-1,-4);
    \draw[thick] (8,1) to [out=180,in=180] (8,-4);

    \draw[thick] (-1,0.5) to [out=180,in=180] (-1,-1) -- (5,-1) to [out=0,in=180] (7,-4.5) -- (8,-4.5);

    \draw[thick] (8,0.5)  to [out=180,in=0] (7,-2) -- (-1,-2) to [out=180,in=180] (-1,-4.5);
        
    \end{tikzpicture}
    \caption{XOR gadget for computing $\oplus_i f_i$ in the garden-hose model. It can be checked directly that if water enters on the top left, it exits on the left if $f_i=0$ and on the right if $f_i=1$. Meanwhile if water enters from the top right, it exits from the bottom right if $f_i=0$ and from the bottom left if $f_i=1$. By wiring such gadgets together for each $f_i$ then, we can compute $\oplus_i f_i$.}
    \label{fig:kpartyGH}
\end{figure}
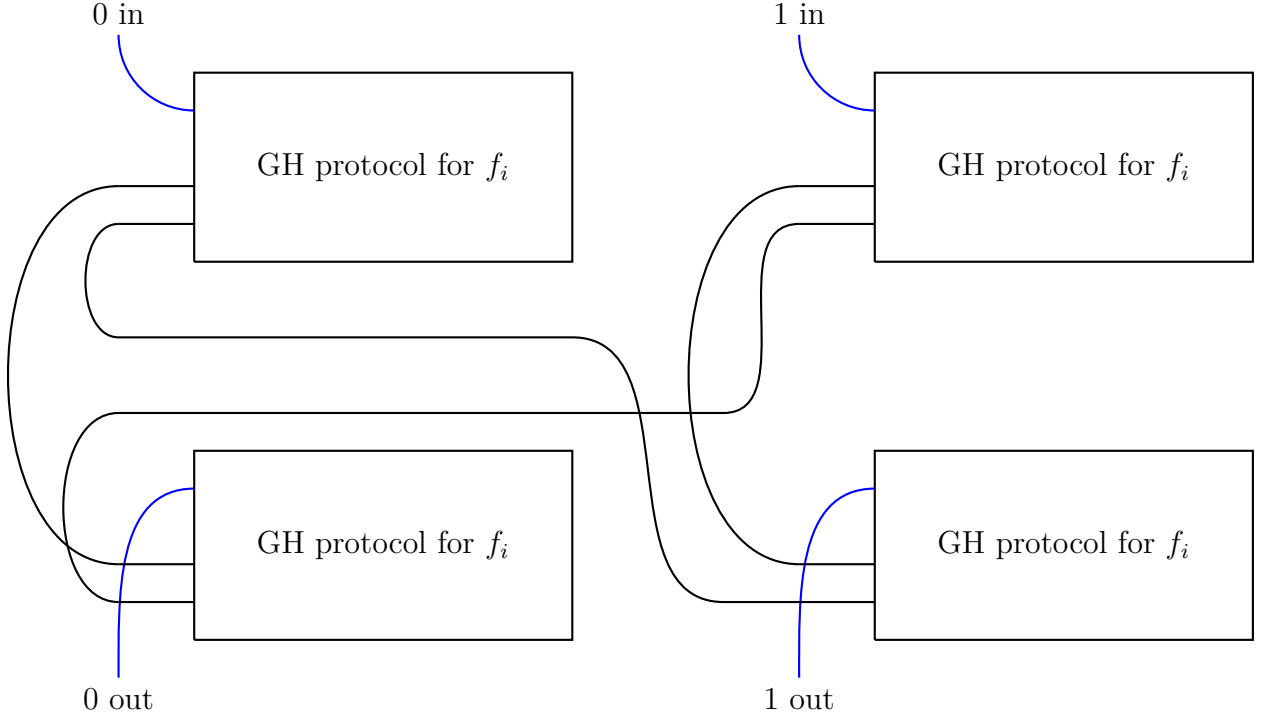

\subsection{A garden-hose gadget}\label{sec:GHcorrectionlemma}

Before stating the precise upper bound and its proof, we first need the following lemma, related to the garden-hose model. 
Specifically, we discuss the garden-hose complexity of implementing a controlled $P^\dagger$ gate instantaneously, up to Pauli corrections, and the garden-hose complexity of the resulting corrections.
The lemma is very similar to one in \cite{speelman2015instantaneous}; the only change is to point out that the proof there continues to apply in the context of the $k>2$ party setting. 

\begin{lemma}\label{lemma:GHgrowth}
    Let $f$ be a k-party function known to all parties. Assume Alice holds a single qubit state $S^{f(x)}\ket{\psi}$, where $x=(x_0,x_1,...,x_{k-1})$ Alice knows $x_0$ and Bob$_i$ knows $x_i$. 
    Then the following two statements hold:
    \begin{enumerate}
        \item There exists an instantaneous protocol (no communication) which uses $2GH_k(f)$ EPR pairs after which Alice holds $X^{g(\hat{x})}Z^{h(\hat{x})}\ket{\psi}$, where $\hat{x}$ depends on $x$ and $2GH(f)$ bits that describe Alice and Bob's measurement outcomes. 
        \item The garden hose complexities of $g$ and $h$ are at most linear in the complexity of $f$, 
        \begin{align}
            GH_k(g) &\leq 4GH_k(f) \nonumber \\
            GH_k(h) &\leq 11GH_k(f)
        \end{align}
    \end{enumerate}
\end{lemma}

\begin{proof}
For the first part, run the garden hose protocol with $S^{f(x)}\ket{\psi}$ as input, have Alice do $S^\dagger$ to the $f=1$ output hose, then wire the outputs to a copy of the protocol. 
The input qubit comes back out of the ``input'' wire of the second protocol, now with no $S$ gate. 
This is illustrated in \cref{fig:GHScontrol}.
Note that we incur Pauli corrections accumulated from all the Bell basis measurements along the path of the qubit. 

Note that in the garden-hose protocol used to apply the conditional $S^\dagger$ there is a sequence of teleportation measurements made, which create possible $X$ and $Z$ corrections. 
Call the bits determining if there is an $X$ correction $b_x^{i,j}$, where $i,j$ label the two EPR pairs involved in the measurement.
Similarly, there are corrections $b_z^{i,j}$. 
Note that not all measurements contribute to these corrections, only those that occur in the unbroken chain of EPR pairs connected to the input state.
Rather than obtain the input state with a $(S^\dagger)^{f(x,y)}$ applied, we actually end up applying, up to a global phase, the operator
\begin{align}
    X^{\sum_{i\in A} b_x^{i,j}}Z^{\sum_{i\in A} b_z^{i,j}} (S^\dagger)^{f(x,y)} X^{\sum_{i\in B} b_x^{i,j}} Z^{\sum_{i\in B} b_z^{i,j}}
\end{align}
where the pairs of indices $(i,j)\in B$ correspond to measurements in the chain that occur before $(S^{-1})^{f(x,y)}$, while pairs $(i,j)\in A$ occur after. 
Using that (again up to a global phase)
\begin{align}
    XZS^{\dagger} &= S^{\dagger}X \nonumber \\
    S^{\dagger}Z &= ZS^{\dagger}
\end{align}
the above becomes
\begin{align}
    X^{g(\hat{x})} Z^{h(\hat{x})} (S^{\dagger})^{f(x,y)} 
\end{align}
where
\begin{align}
    g(\hat{x}) &= \sum_{(i,j)\in A\cup B} b_x^{i,j}, \nonumber \\
    h(\hat{x}) &= \sum_{(i,j)\in A\cup B} b_z^{i,j} + f(x,y) \sum_{(i,j)\in B} b_x^{i,j}.
\end{align}
Thus to compute $g$, we just need to compute the parity of all of the $b_x^{i,j}$ that occur in the chain. 
Note that which measurements are actually a part of the chain depends on $x$, so this is a function of the original input $x$ as well as the measurement outcomes $b_x^{(i,j)}$.
The function $h$ is somewhat more involved, in particular there is an additional correction based on the $b_{x}^{(i,j)}$ for measurements that occur before the conditional $S^{\dagger}$. 

Let's begin with designing a garden-hose protocol to compute $g(\hat{x})$. 
To do this, we create a ``rail'', consisting of two EPR pairs, one for each EPR pair in the initial protocol. 
Then, we connect subsequent rails in the ordering defined by the sequence of EPR pairs used in the original protocol. 
We connect the rails end to end if $b_x^{i,j}=0$, and we connect them cross-wise if $b_x^{i,j}=1$. 
Thus after running over all $(i,j)$, the input is crossed if the parity of the $b_x^{i,j}$ is odd, and left unchanged if the parity is even. 
This protocol uses twice the EPR pairs used in the protocol for applying $(S^{-1})^{f(x,y)}$, which itself was $2GH(f)$, so the cost is $4GH(f)$. 

Now we consider the function $h(\hat{x})$. 
We need a somewhat more involved protocol that treats EPR pairs before and after the conditional $S^{\dagger}$ differently, and accounts for the value of $f(x,y)$. 
To do this, we first run a garden-hose protocol to compute $f(x,y)$, then feed the two output hoses into two different subsequent garden-hose protocols.
The $f=0$ pipe is input to a protocol computing the parity of just the $Z$ corrections. 
We do this using the ``rail'' construction, just as in computing $g(\hat{x})$. 
The $f=1$ pipe is input to a similar rail protocol, which flips the rails if $b_{x}^{(i,j)}\oplus b_z^{(i,j)}=1$ for pipes occurring before the $S^{\dagger}$, and flips the pipes after the $S^{\dagger}$ if $b_z^{(i,j)}=1$.  
The garden-hose complexity of this protocol is composed of:
\begin{itemize}
    \item The complexity of computing $f(x,y)$, in a way that uses just two spilling pipes, which is $3GH_k(f)$.
    \item The complexity of computing the $Z$ corrections only, in the sub-protocol that is used when $f(x,y)=0$. This is $4GH_k(f)$, where the 4 comes from using the rail construction to double the pipes in the initial protocol, which itself was the protocol that involved computing $f$, applying $S^{\dagger}$, then running the protocol for $f$ in reverse.
    \item The complexity of computing the parity of the $b_{x}^{(i,j)}\oplus b_z^{(i,j)}$ for the first part of the protocol (before $S^{\dagger}$) along with the parity of the $b_z^{(i,j)}$ in the later part of the protocol. This is $4GH_k(f)$ again.  
\end{itemize}
In total then the garden hose complexity of $h(\hat{x})$ is $11GH_k(f)$. 
\end{proof}

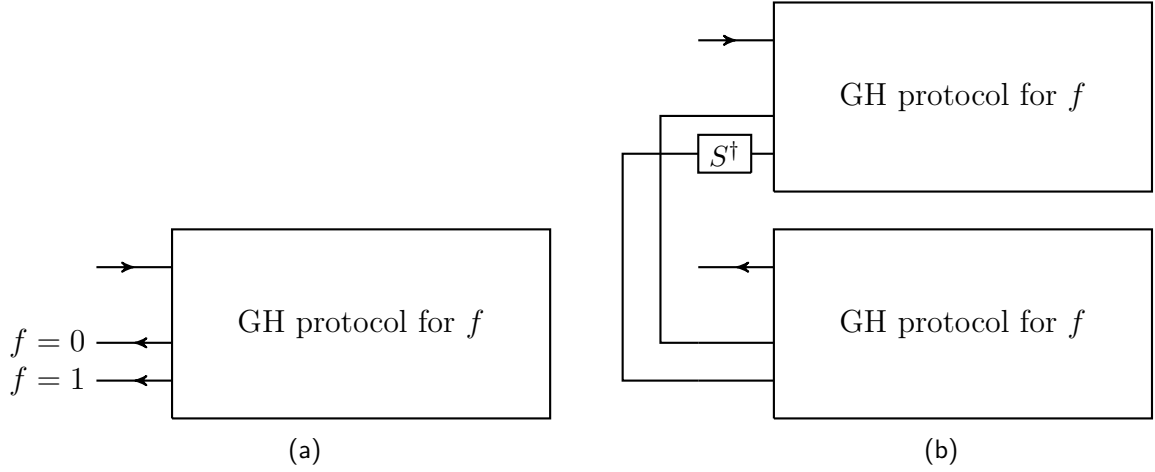
\begin{figure}
    \centering
    \begin{subfigure}{0.45\textwidth}
    \begin{tikzpicture}

    \draw[thick] (0,0) -- (0,2.5) -- (5,2.5) -- (5,0) -- (0,0);
    \draw[thick,mid arrow] (-1,2) -- (0,2);

    \draw[thick,mid arrow] (0,1) -- (-1,1);
    \node[left] at (-1,1) {$f=0$};
    \draw[thick,mid arrow] (0,0.5) -- (-1,0.5);
    \node[left] at (-1,0.5) {$f=1$};

    \node at (2.5,1.25) {GH protocol for $f$};
        
    \end{tikzpicture}
    \caption{}
    \end{subfigure}
    \begin{subfigure}{0.47\textwidth}
    \begin{tikzpicture}

    \draw[thick] (0,0) -- (0,2.5) -- (5,2.5) -- (5,0) -- (0,0);
    \draw[thick,mid arrow] (-1,2) -- (0,2);

    \draw[thick] (0,1) -- (-1.5,1) -- (-1.5,1-3) -- (-1,1-3);
    \draw[thick] (0,0.5) -- (-2,0.5) -- (-2,0.5-3) -- (-1,0.5-3);

    \draw[thick,fill=white] (-0.3,0.25) -- (-0.3,0.75) -- (-1,0.75) -- (-1,0.25) -- (-0.3,0.25);
    \node at (-0.65,0.5) {\small{$S^\dagger$}};

    \node at (2.5,1.25) {GH protocol for $f$};

    \draw[thick] (0,0-3) -- (0,2.5-3) -- (5,2.5-3) -- (5,0-3) -- (0,0-3);
    \draw[thick,mid arrow] (0,2-3) --  (-1,2-3);

    \draw[thick] (0,1-3) -- (-1,1-3);
    \draw[thick] (0,0.5-3) -- (-1,0.5-3);

    \node at (2.5,1.25-3) {GH protocol for $f$};
                
    \end{tikzpicture}
    \caption{}
    \end{subfigure}
    \caption{a) $k$-party garden-hose protocol for a function $f$. The water enters in the top left pipe, is redirected through a sequence of pipes represented abstractly as the white rectangle, then exits through one of two pipes, indicating the value of $f(x,y)$. The open input and both open outputs are held by Alice, while the remaining parties are involved in the intermediate steps. b) Protocol for applying a $(S^\dagger)^{f(x,y)}$. The input state is entered into a garden hose protocol that computes $f$, then $S^\dagger$ is applied only to the output pipe indicating $f(x,y)=1$. Another instance of the protocol for $f$ is then run, and the corresponding output pipes of the two protocols are connected. As a result water always runs out a single fixed pipe.}
    \label{fig:GHScontrol}
\end{figure}

\subsection{Instantaneous \texorpdfstring{$k$}{k}-party computations from \texorpdfstring{$T$}{T}-depth}\label{sec:T-depth}

We are ready to prove our upper bound on the entanglement cost of implementing a unitary based on the $T$-depth. 
Before delving into the detailed proof, we give a heuristic understanding of where the dominant scaling of the entanglement cost comes from. 
The protocol involves first having all of the Bobs perform Bell basis measurements using EPR pairs shared with Alice, as if teleporting their states to her. 
This gives the full state in Alice's lab, up to Pauli corrections which are determined by a distributed set of measurement outcomes. 
Alice then applies her first Clifford layer, then applies the first layer of $T$ gates. 
We can conjugate the Pauli's through the Clifford to give new Pauli's, then through the $T$ gates using the identities
\begin{align}
    ZT &=TZ, \nonumber \\
    TX &=SXT,
\end{align}
which hold up to a global phase. 
This leaves us in a state where one layer of the unitary has been applied, up to $S$ corrections. 
We then use \cref{lemma:GHgrowth} to correct the $S$ gates. 
There is an entanglement cost to doing this set by the garden-hose complexity of the function determining if there is an $S$ gate or not.
After correcting the $S$ gates as needed, we begin again in the situation we started in, trying to apply a Clifford+$T$ layer to a state with unknown Pauli corrections applied to it. 
We repeat the above procedure to apply the next layer. 
After each layer, the garden-hose complexity of the needed $S$ corrections grows, and this growth controls the entanglement cost. 

To capture the form of the $S$ corrections more precisely and understand how their garden-hose complexity grows, define $g_{i,j}=1$ if there is an $X$ correction on the $j$th input wire and $0$ otherwise, along with $g_{i+1,j}=1$ if an $X$ correction appears on the $j$th output wire after conjugation. 
Similarly, we define $h_{i,j}$ and $h_{i+1,j}$ to be 1 to indicate a $Z$ correction on the input or output $j$th wire, respectively. 
Then, we can see that the output wire functions are related to the input wire functions by
\begin{align}
    g_{i+1,k}&=\bigoplus_{j\in S_{g,k}} g_{i,j} \oplus \bigoplus_{j\in S_{g,k}'}h_{i,j},\nonumber \\
    h_{i+1,k}&=\bigoplus_{j\in S_{h,k}} g_{i,j} \oplus \bigoplus_{j\in S_{h,k}'}h_{i,j}.
\end{align}
The subsets $S_{g,k}^{(\prime)}$ and $S_{h,k}^{(\prime)}$ depend on the choice of Clifford and the wire $k$ being considered.
For Clifford with backward light cones of size at most $\ell$, we have that these sets are not larger than $\ell$. 
From \cref{lemma:XORlemma}, we know how the garden-hose complexity of the XOR of many functions behaves, and in particular we can bound it by something of order $\ell$ times the worst-case garden-hose complexity of the $g_{i,j}$ and $h_{i,j}$. 
The garden-hose complexity of the worst single qubit correction at layer $i+1$, call it $t_{i+1}$, then is related to the complexity at the previous layer by $t_{i+1}\lesssim \ell \,t_i$.
An added complication is that these Pauli corrections move through the $T$ gates at this layer to give $S$ gates, and then to correct the $S$ gates we apply \cref{lemma:GHgrowth}. 
This increases the garden-hose complexity of the Pauli corrections on that wire, but only by a constant factor that contributes to the value of $K$.
It is the XOR functions determined by the choice of Clifford that give the dominant scaling of the entanglement cost. 

Finally, to solve the recursive relation $t_{i+1}\lesssim t_i \ell$ we need to know the garden-hose complexity at the $0$th layer, which means the garden-hose complexity of the Pauli corrections right after the initial teleportation into Alice's lab. 
These have constant garden-hose complexity, since they depend on just one bit held in one of the Bobs' labs. 
After the first Clifford, denoted $C_0^\ell$ then, the garden-hose complexity is at most $\ell$, since the functions $g_{1,i}$, $h_{1,i}$ are XOR's of $\ell$ bits. 
However, it is possible some of those bits are held together in one lab, in which case the parity functions $g_{1,i}$, $h_{1,i}$ only depend on a parity of a subset of them, which the local party can pre-compute. 
For instance if we have $k$-parties, $g_{1,i}$, $h_{1,i}$ are XOR functions of at most $k$ distributed bits. 
This means $t_1=O( \min\{k,\ell\})$, and hence $t_{d}\lesssim \min\{k,\ell\} (K\ell)^{d-1}$. 
Letting $\ntot$ be the total number of qubits the unitary acts on, the entanglement cost for each layer is then bounded by $\ntot$, the maximal number of possible $S$ corrections at each layer, times $t_i$, the worst-case correction complexity. 
The total entanglement cost is then the sum of the cost for each layer, 
\begin{align}
    E\leq \sum_{i=1}^d t_{i} \ntot \leq O((K\ell)^{d-1}\cdot \ntot \cdot \min\{k,\ell\})
\end{align}
which is the claimed scaling. 
We give a formal proof and more detailed accounting in the theorem proof below. 

\begin{theorem}\label{thm:Tdepth}
    Given a unitary $U_{AB_1...B_{k-1}}$ that can be implemented in a Clifford+$T$ decomposition using a circuit of $T$-depth $d$, we have that $U$ can be computed instantaneously using
    \begin{align}
        E(U)\leq O((K\ell)^{d-1}\cdot \ntot \cdot \min\{k,\ell\})
    \end{align}
    EPR pairs.  
\end{theorem}
\begin{proof} We first have each of the Bobs 'teleport'\footnote{By this we mean have each Bob measure their system in the Bell basis along with one end of a maximally entangled state, with the other end held by Alice.} their systems $B$ to Alice, who then holds
    \begin{align}
        X^{\vec{g}_{0}(y)}Z^{\vec{h}_0(y)}\ket{\psi}_{AB}
    \end{align}
    where, less succinctly, we mean 
    \begin{align}
        X^{\vec{g}_0(y)}&=X_1^{g_{0,1}(y)}...X_n^{g_{0,n}(y)}\nonumber \\
        Z^{\vec{h}_0(y)}&=Z_1^{h_{0,1}(y)}...Z_n^{h_{0,n}(y)}
    \end{align}
    where $X_i$ and $Z_i$ act on the $i$th qubit. 
    Note that the entries of both $\vec{h}$ and $\vec{g}$ all have constant garden-hose complexity, since they are functions of single bits held in one of the Bobs labs. 
    
    This will serve as our base case in an inductive argument. 
    We induct on the level $i$, and assume Alice holds the state
    \begin{align}           X^{\vec{g}_i(x,y)}Z^{\vec{h}_i(x,y)}\bar{T}_iC_i...\bar{T}_1C_1\ket{\psi}_{AB}.
    \end{align}
    where Alice holds $x$ and Bob holds $y$, and the entries of $\vec{g}_i$ and $\vec{h}_i$ have known garden-hose complexities. 
    Define
    \begin{align}
        t_i=\max\{\max_j\{GH(g_{i,j})\},\max_j\{ GH(h_{i,j})\}\}.
    \end{align}
    In words $t_i$ is the worst-case garden-hose complexity of any single $X$ or $Z$ correction in the $i$th layer. 
    We have from above that $t_0=2$. 
    
    To induct have Alice apply $\bar{T}_{i+1}C_{i+1}$, obtaining 
    \begin{align}        \bar{T}_{i+1}C_{i+1}X^{\vec{g}_i(x,y)}Z^{\vec{h}_i(x,y)}\bar{T}_iC_i...\bar{T}_1C_1\ket{\psi}_{AB} = \bar{S}^{\vec{f}_i(x,y)}X^{\vec{g}_i'(x,y)}Z^{\vec{h}_i'(x,y)}\bar{T}_{i+1}C_{i+1}...\bar{T}_1C_1\ket{\psi}_{AB} \nonumber
    \end{align}
    Then, we use the procedure of \cref{lemma:GHgrowth} to undo the phase gates, obtaining
    \begin{align}            X^{\vec{g}_{i}''(x,y)\oplus\vec{g}_i'(x,y)}Z^{\vec{h}_i''(x,y)\oplus\vec{h}_i'(x,y)}\bar{T}_{i+1}C_{i+1}...\bar{T}_1C_1\ket{\psi}_{AB}. 
    \end{align}
    The functions $\vec{g}'_{i}, \vec{h}'_{i}$ arise from commuting the $X, Z$ operators through $\bar{T}_{i+1}C_{i+1}$, while the $\vec{g}''_{i}, \vec{h}''_{i}$ operators appear when correcting the phase gates.
    The singly-primed operators are of the form
    \begin{align}\label{eq:parityform}
        g'_{i,j}&=\bigoplus_{l\in S_{g,j}}{g}_{i,l}\oplus\bigoplus_{k\in S'_{g,j}}{h}_{i,k} \nonumber \\
        h'_{i,j}&=\bigoplus_{l\in S_{h,j}}{g}_{i,l}\oplus\bigoplus_{k\in S'_{h,j}}{h}_{i,k}
    \end{align}
    where the subsets $S_{g/h,j}^{(\prime)}$ depend on the Clifford, and have size at most $\ell$. 
    
    The functions $\vec{g}''_{i}, \vec{h}''_{i}$ appear when undoing the $S^{\vec{f}(x,y)}$ operator, which we do using the procedure in \cref{lemma:GHgrowth}. 
    We are also provided with upper bounds on the garden-hose complexity of these functions from that lemma. 
    We want to determine the garden-hose complexity of $g_{i+1,j}=g_{i,j}''\oplus g_{i,j}'$ and $h_{i+1,j}=h_{i,j}''\oplus h_{i,j}'$. 
    Starting with $g_{i+1,j}$, we have
    \begin{align}
        GH(g_{i+1,j})&=GH\left(\bigoplus_{l\in S_{g,j}}{g}_{i,l}\oplus\bigoplus_{k\in S'_{g,j}}{h}_{i,k}\oplus g_{i,j}'' \right)\nonumber \\
        &\leq 4\left(\sum_{l}GH\left(g_{i,l} \right)+\sum_k GH(h_{i,k}) + GH(g''_{i,j})\right) \nonumber \\
        &\leq  4\left(\sum_{l}GH\left(g_{i,l} \right)+\sum_k GH(h_{i,k}) + 4GH(f_{i,j})\right) \nonumber \\
        &\leq  4\left(\ell t_i+ \ell t_i + 4GH(f_{i,j})\right)
    \end{align}
    where the first inequality uses \cref{lemma:XORlemma} (the XOR lemma), the second line uses that $GH(g''_{i,j})\leq 4GH(f_{i,j})$ which comes from \cref{lemma:GHgrowth}, and the last uses the definition of $t_i$. 
    
    It remains to bound the garden-hose complexity of $f_{i,j}$. 
    Notice that since $f_{i,j}$ is itself a parity function of the $g_{i,j}$ and $h_{i,j}$ (it is of the form \eqref{eq:parityform}), so again its garden hose complexity is at most $4\ell t_i$ by a use of the XOR lemma.
    Overall then this gives
    \begin{align}\label{eq:gupper}
        GH(g_{i+1,j})&\leq 4(\ell t_i+\ell t_i + 4\cdot 4\ell t_i) = 72\ell t_i
    \end{align}
    An upper bound can be determined for $GH(h_{i+1,j})$ in a similar way. 
    The only difference is that where before we used $GH(g''_{i,j})\leq 4GH(f_{i,j})$, we now need $GH(h''_{i,j})\leq 11GH(f_{i,j})$, which is given in \cref{lemma:GHgrowth}.
    This changes the constant but gives a similar upper bound, 
    \begin{align}\label{eq:hupper}
        GH(h_{i+1,j})&\leq 184 \ell t_i.
    \end{align}
    Using equations \eqref{eq:gupper} and \eqref{eq:hupper}, we get that
    \begin{align}
        t_{i+1}\leq 184 \ell t_i.
    \end{align}
    Our numerical constant is not optimal, but we lose the optimal constants in favour of a simpler presentation. 
    We will write $t_{i+1}=K\ell t_i$ to summarize the above. 
    This relation and our earlier computation showing that $t_1=\min\{\ell,k\}$ is solved by
    \begin{align}
        t_d\leq O((K\ell )^{d-1} \min\{\ell,k\}).
    \end{align}
    Summing over the corrections at each of the $d$ layers, which each involve $\ntot$ qubits, leads to the claimed upper bound. 
\end{proof}

\section{Newman's Theorem for PSM}

The following theorem shows that the randomness complexity in the PSM model is never much larger than the communication complexity. 
The proof is straightforward and similar to the well-known Newman's theorem in communication complexity \cite{newman1991private}; however we have not seen it recorded explicitly for PSM, so we give a proof here for completeness. 
A similar statement for the related conditional disclosure of secrets setting was proven in \cite{applebaum2021placing}. 

\begin{theorem} \textbf{(Newman's theorem for PSM)} Consider an $\epsilon$-correct and $\delta$-secure PSM protocol for function $f$, with communication complexity $c$. 
Then, there exists an $\epsilon+\delta'$-correct and $\delta+2\delta'$-secure PSM protocol for $f$ using communication complexity $c$ and randomness complexity $2c+O(\log\left(c+n\right) + \log(1/\delta'))$.
\end{theorem}
\begin{proof}\, Let $D$ be the distribution from which the players sample their randomness (independently of their inputs). 
Sample $r_1,\ldots,r_K\sim D$ where we take $K=2^c\cdot (c+n)\cdot \log(1/\delta')/\epsilon^2$. 
We will think of $r_1,\ldots,r_K$ being the common source of randomness for all inputs and in the new protocol, the players will instead sample $k\sim [K]$ and run the original protocol with $r_k$ as the randomness. 
We will now analyze the correctness and privacy of this protocol. 

\paragraph*{Privacy.} Fix any inputs $x,y$. 
Now, each fixed random string $r_i$ induces a fixed transcript $m_{x,y}(r_i)$ and we consider the distribution $M'_{x,y}$ on transcripts obtained by sampling $i\sim[K]$ and outputting $m_{x,y}(r_i)$. 
Compare it to the ideal distribution $M_{x,y}$ on transcripts obtained by sampling a random $r\sim D$ and outputting $m_{x,y}(r)$. 
For any $\tau$ in the support of the message space, we want that with high probability, the probabilities assigned to $\tau$ by $M_{x,y}$ and by $M’_{x,y}$ differ very little — by at most $t=\delta'/2^c$. 
By Hoeffding’s inequality, these probabilities differ by at least $t$ with probability at most $\exp(-t^2K)$.
This means, by a union bound over all possible transcripts $\tau$, that the total variation distance between $M_{x,y}$ and $M_{x,y}'$, given by
\begin{align}
    \frac{1}{2}\sum_\tau |M_{x,y}[\tau]-M_{x,y}'[\tau]| 
\end{align}
is at most $t\cdot 2^{c}/2$ with probability at least $1- \exp(-t^2K)\cdot 2^c$. 
Further, we want that these distributions are $\delta'$-close for all input pairs $x,y$, which (by the union bound again) occurs with probability at least $1-\exp(-t^2K)\cdot 2^c\cdot 2^n$. 
If this happens, then it would follow from the Triangle Inequality that the new protocol is $(\delta+2\delta')$-private. 
Setting $t=\delta'/2^c$ so that the total variation distance is at most $\delta'$, we obtain that the bad event has probability at most 
$\exp(-\delta'^2K/2^{2c})\cdot 2^c\cdot 2^n$. 
We want this to be strictly smaller than 1, so that by linearity of expectation, we can fix some random strings $r_1,\ldots,r_K$ such that the distributions $M'_{x,y}$ and $M_{x,y}$ are $\delta'$-close in total variation distance for all $x,y$. 
To achieve this bound, we only need to set $K=\Theta(2^{2c}\cdot (c+n)/\delta'^2)$. 

\paragraph*{Correctness.} The correctness of this protocol follows immediately from the guarantee from earlier that the distribution $M'_{x,y}$ is $\delta'$-close to $M_{x,y}$ for all $x,y$ -- indeed, the referee's output only depends on the distribution of the transcript, and since the original protocol is $\epsilon$-correct, the new protocol is $\epsilon+\delta'$-correct.
\end{proof}

\bibliographystyle{unsrtnat}
\bibliography{biblio}

@inproceedings{barrington1986bounded,
  title={Bounded-width polynomial-size branching programs recognize exactly those languages in {NC}},
  author={Barrington, David A},
  booktitle={Proceedings of the eighteenth annual ACM symposium on Theory of computing},
  pages={1--5},
  year={1986}
}

@inproceedings{feige1994minimal,
  title={A minimal model for secure computation},
  author={Feige, Uri and Killian, Joe and Naor, Moni},
  booktitle={Proceedings of the twenty-sixth annual ACM symposium on Theory of computing},
  pages={554--563},
  year={1994}
}

@inproceedings{ishai1997private,
  title={Private simultaneous messages protocols with applications},
  author={Ishai, Yuval and Kushilevitz, Eyal},
  booktitle={Proceedings of the Fifth Israeli Symposium on Theory of Computing and Systems},
  pages={174--183},
  year={1997},
  organization={IEEE}
}

@inproceedings{ball2020complexity,
  title={On the complexity of decomposable randomized encodings, or: How friendly can a garbling-friendly PRF be?},
  author={Ball, Marshall and Holmgren, Justin and Ishai, Yuval and Liu, Tianren and Malkin, Tal},
  booktitle={11th Innovations in Theoretical Computer Science Conference (ITCS 2020)},
  pages={86--1},
  year={2020},
  organization={Schloss Dagstuhl--Leibniz-Zentrum f{\"u}r Informatik}
}

@article{applebaum2020communication,
  title={The communication complexity of private simultaneous messages, revisited},
  author={Applebaum, Benny and Holenstein, Thomas and Mishra, Manoj and Shayevitz, Ofer},
  journal={Journal of Cryptology},
  volume={33},
  number={3},
  pages={917--953},
  year={2020},
  publisher={Springer}
}

@article{girish2025comparing,
  title={Comparing classical and quantum conditional disclosure of secrets},
  author={Girish, Uma and May, Alex and Orshansky, Leo and Waddell, Chris},
  journal={arXiv preprint arXiv:2505.02939},
  year={2025}
}

@article{kawachi2021communication,
  title={Communication complexity of private simultaneous quantum messages protocols},
  author={Kawachi, Akinori and Nishimura, Harumichi},
  journal={arXiv preprint arXiv:2105.07120},
  year={2021}
}

@inproceedings{ball2022note,
  title={A note on the complexity of private simultaneous messages with many parties},
  author={Ball, Marshall and Randolph, Tim},
  booktitle={3rd Conference on Information-Theoretic Cryptography (ITC 2022)},
  pages={7--1},
  year={2022},
  organization={Schloss Dagstuhl--Leibniz-Zentrum f{\"u}r Informatik}
}

@book{KN, place={Cambridge}, title={Communication Complexity}, publisher={Cambridge University Press}, author={Kushilevitz, Eyal and Nisan, Noam}, year={1996}}

@article{girish2025magic,
  title={Magic and communication complexity},
  author={Girish, Uma and May, Alex and Parham, Natalie and Yuen, Henry},
  journal={arXiv preprint arXiv:2510.07246},
  year={2025}
}

@article{asadi2024rank,
  title={Rank lower bounds on non-local quantum computation},
  author={Asadi, Vahid R and Culf, Eric and May, Alex},
  journal={arXiv preprint arXiv:2402.18647},
  year={2024}
}

@article{speelman2015instantaneous,
  title={Instantaneous non-local computation of low {T}-depth quantum circuits},
  author={Speelman, Florian},
  journal={arXiv preprint arXiv:1511.02839},
  year={2015}
}

@inproceedings{buhrman2013garden,
  title={The garden-hose model},
  author={Buhrman, Harry and Fehr, Serge and Schaffner, Christian and Speelman, Florian},
  booktitle={Proceedings of the 4th conference on Innovations in Theoretical Computer Science},
  pages={145--158},
  year={2013}
}

@article{klauck2014new,
  title={New bounds for the garden-hose model},
  author={Klauck, Hartmut and Podder, Supartha},
  journal={arXiv preprint arXiv:1412.4904},
  year={2014}
}

@article{allerstorfer2024relating,
  title={Relating non-local quantum computation to information theoretic cryptography},
  author={Allerstorfer, Rene and Buhrman, Harry and May, Alex and Speelman, Florian and Lunel, Philip Verduyn},
  journal={Quantum},
  volume={8},
  pages={1387},
  year={2024},
  publisher={Verein zur F{\"o}rderung des Open Access Publizierens in den Quantenwissenschaften}
}

@article{asadi2025conditional,
  title={Conditional disclosure of secrets with quantum resources},
  author={Asadi, Vahid R and Kuroiwa, Kohdai and Leung, Debbie and May, Alex and Pasterski, Sabrina and Waddell, Chris},
  journal={Quantum},
  volume={9},
  pages={1885},
  year={2025},
  publisher={Verein zur F{\"o}rderung des Open Access Publizierens in den Quantenwissenschaften}
}

@article{GERTNER2000592,
title = {Protecting Data Privacy in Private Information Retrieval Schemes},
journal = {Journal of Computer and System Sciences},
volume = {60},
number = {3},
pages = {592-629},
year = {2000},
issn = {0022-0000},
doi = {https://doi.org/10.1006/jcss.1999.1689},
url = {https://www.sciencedirect.com/science/article/pii/S0022000099916896},
author = {Yael Gertner and Yuval Ishai and Eyal Kushilevitz and Tal Malkin}
}

@inproceedings{buhrman2001communication,
  title={Communication complexity lower bounds by polynomials},
  author={Buhrman, Harry and de Wolf, Ronald},
  booktitle={Proceedings 16th Annual IEEE Conference on Computational Complexity},
  pages={120--130},
  year={2001},
  organization={IEEE}
}

@inproceedings{yao1993quantum,
  title={Quantum circuit complexity},
  author={Yao, A Chi-Chih},
  booktitle={Proceedings of 1993 IEEE 34th Annual Foundations of Computer Science},
  pages={352--361},
  year={1993},
  organization={IEEE}
}

@book{o2014analysis,
  title={Analysis of boolean functions},
  author={O'Donnell, Ryan},
  year={2014},
  publisher={Cambridge University Press}
}

@article{grolmusz1997power,
  title={On the power of circuits with gates of low L1 norms},
  author={Grolmusz, Vince},
  journal={Theoretical computer science},
  volume={188},
  number={1-2},
  pages={117--128},
  year={1997},
  publisher={Elsevier}
}

@article{newman1991private,
  title={Private vs. common random bits in communication complexity},
  author={Newman, Ilan},
  journal={Information processing letters},
  volume={39},
  number={2},
  pages={67--71},
  year={1991}
}

@article{alicki2004continuity,
  title={Continuity of quantum conditional information},
  author={Alicki, Robert and Fannes, Mark},
  journal={Journal of Physics A: Mathematical and General},
  volume={37},
  number={5},
  pages={L55--L57},
  year={2004}
}

@article{winter2016tight,
  title={Tight uniform continuity bounds for quantum entropies: conditional entropy, relative entropy distance and energy constraints},
  author={Winter, Andreas},
  journal={Communications in Mathematical Physics},
  volume={347},
  number={1},
  pages={291--313},
  year={2016},
  publisher={Springer}
}

@article{applebaum2021placing,
  title={Placing conditional disclosure of secrets in the communication complexity universe},
  author={Applebaum, Benny and Vasudevan, Prashant Nalini},
  journal={Journal of Cryptology},
  volume={34},
  number={2},
  pages={11},
  year={2021},
  publisher={Springer}
}

@article{mottonen2004quantum,
  title={Quantum circuits for general multiqubit gates},
  author={M{\"o}tt{\"o}nen, Mikko and Vartiainen, Juha J and Bergholm, Ville and Salomaa, Martti M},
  journal={Physical review letters},
  volume={93},
  number={13},
  pages={130502},
  year={2004},
  publisher={APS}
}

@article{kim2025catalytic,
  title={Catalytic $ z $-rotations in constant $ T $-depth},
  author={Kim, Isaac H},
  journal={arXiv preprint arXiv:2506.15147},
  year={2025}
}

@article{kim2025any,
  title={Any {C}lifford+{T} circuit can be controlled with constant {T}-depth overhead},
  author={Kim, Isaac H and Laakkonen, Tuomas},
  journal={arXiv preprint arXiv:2512.24982},
  year={2025}
}

@article{kitaev1997quantum,
  title={Quantum computations: algorithms and error correction},
  author={Kitaev, A Yu},
  journal={Russian Mathematical Surveys},
  volume={52},
  number={6},
  pages={1191--1249},
  year={1997}
}

@book{nielsen2010quantum,
  title={Quantum computation and quantum information},
  author={Nielsen, Michael A and Chuang, Isaac L},
  year={2010},
  publisher={Cambridge university press}
}

@misc{gidney_x_1936285631359197210,
  author       = {Gidney, Craig},
  title        = {Post on {X} ({T}witter)},
  howpublished = {\url{https://x.com/CraigGidney/status/1936285631359197210}},
  note         = {Accessed: 2026-02-12},
}

\end{document}